\documentclass[a4paper,USenglish,cleveref,autoref,thm-restate,
  numberwithinsect]{lipics-v2021}
\hideLIPIcs
\nolinenumbers

\usepackage{amsmath,amssymb,mathtools}
\usepackage{booktabs}
\usepackage[classfont=bold]{complexity}
\usepackage{float}
\usepackage{tabularx}
\usepackage{tikz}
\usetikzlibrary{arrows.meta,positioning,calc}
\usepackage{needspace}

\newif\ifanon
\title{NEXP-Completeness of Existential Presburger Arithmetic with Divisibility}
\titlerunning{NEXP-Completeness and Exponential Coefficient Growth for EPAD}
\ifanon
\author{Anonymous Author(s)}{Affiliation withheld for review}{}{}{}
\authorrunning{Anonymous Author(s)}
\Copyright{}
\hypersetup{pdfauthor={Anonymous Author(s)}}
\else
\author{Ignacio Barros}
  {University of Antwerp, Belgium}
  {ignacio.barros@uantwerp.be}
  {https://orcid.org/0000-0002-7729-9413}
  {}
\author{Micha{\"e}l Cadilhac}
  {DePaul University, USA}
  {michael@cadilhac.name}
  {https://orcid.org/0000-0001-9828-9129}
  {}
\author{Guillermo A. P{\'e}rez}
  {University of Antwerp -- Flanders Make, Belgium}
  {guillermo.perez@uantwerp.be}
  {https://orcid.org/0000-0002-1200-4952}
  {}
\authorrunning{I. Barros, M. Cadilhac, and G. A. P{\'e}rez}
\Copyright{Ignacio Barros, Micha{\"e}l Cadilhac, and Guillermo A. P{\'e}rez}
\fi
\category{}
\relatedversion{}
\supplement{}
\ifanon
\funding{}
\acknowledgements{}
\else
\funding{This work was supported by the Belgian FWO-FNRS ``SynthEx''
  (G0AH524N) project and by the FWO project number G0D9323N.}
\acknowledgements{We thank Alessio Mansutti for feedback on the normalization
  presentation and the connection to word equations with length constraints.}
\fi

\newcommand{\Z}{\mathbb Z}
\newcommand{\N}{\mathbb N}
\newcommand{\Npos}{\mathbb N_{>0}}
\newcommand{\lcm}{\operatorname{lcm}}
\newcommand{\val}{\operatorname{val}}
\newcommand{\word}{\operatorname{word}}
\newcommand{\enc}{\operatorname{enc}}
\newcommand{\Div}{\operatorname{Div}}
\newcommand{\Scale}{\operatorname{Scale}}
\renewcommand{\vec}[1]{\boldsymbol{#1}}
\newcommand{\bits}{\{0,1\}}
\newcommand{\polyALC}{\ComplexityFont{polyALC}}
\newcommand{\expALC}{\ComplexityFont{expALC}}
\newcommand{\polyALCstd}{\ComplexityFont{polyALC}^{\mathrm{std}}}
\newcommand{\FEPAD}{\ComplexityFont{FEPAD}}
\providecommand{\FPSPACE}{\ComplexityFont{FPSPACE}}
\providecommand{\FEXP}{\ComplexityFont{FEXP}}
\newcommand{\len}{\operatorname{len}}
\providecommand{\poly}{\operatorname{poly}}
\newcommand{\MergeAbs}{\textnormal{merge-absorptive}}

\theoremstyle{definition}
\newcommand{\exampleqed}{\unskip\nobreak\hfill\(\triangleleft\)}
\newcommand{\proofparagraph}[1]{\par\smallskip\noindent\textbf{#1.}\ }
\BeforeBeginEnvironment{example}{\Needspace{8\baselineskip}}
\newtheorem*{overviewdef}{Definition}

\begin{document}
\pagenumbering{roman}
\pagestyle{empty}
\maketitle

\begin{abstract}
  We prove that satisfiability for existential Presburger arithmetic with
  divisibility (EPAD) is \(\NEXP\)-hard.  Together with the known \(\NEXP\) upper
  bound, this establishes \(\NEXP\)-completeness.  This settles a 50-year old open
  question, improving on Lipshitz's 1981 \NP-hardness lower bound which had been
  conjectured to be tight.  We reduce from the \(\NEXP\)-complete succinct version of Exact-2-in-4-SAT.
  In this problem, a clause is satisfied when exactly two of its four literal occurrences are true, and a polynomial-size circuit, given a clause index in binary, returns the corresponding clause of an exponentially large formula.
  A central difficulty is to
  represent, with a polynomial-size EPAD formula, the exponentially long
  integers arising in the encoding.

  Any requested bit of two constants required by the reduction can be computed
  in polynomial space.  To turn this fact into short arithmetic definitions, we
  introduce fixed-width arithmetic logic circuits (ALCs), whose gates perform addition, multiplication,
  and bit shifts.  Polynomial-time uniform polynomial-size ALC families compute
  exactly the nonnegative-integer functions of polynomial-space transducers
  (\(\FPSPACE\)), and polynomial-size EPAD formulas uniquely define their
  outputs.  Allowing exponentially many gates while retaining polynomial-time local uniformity, under which a gate's type and predecessors are recovered from its binary index, gives exactly the analogous deterministic exponential-time function class \(\FEXP\).
  We also prove that every \(\FEXP\) function
  has a polynomial-size EPAD definition that is functional in its output.

  As further consequences, we obtain \(\NEXP\)-hardness for positive Boolean
  combinations of word equations and Presburger length constraints over a fixed
  two-letter alphabet, already when every variable denotes a nonempty word, and
  that reachability for parametric one-counter automata is \(\NEXP\)-complete.

  We also study the elimination of divisibility constraints by enumerating their possible quotients, when bounded.
  Earlier work identified an \(\NP\) fragment in which the variables can be ordered so that divisibility dependencies always move forward through the order.
  We introduce a different fragment, called \MergeAbs{}, in which bounded-quotient divisibilities can successively merge groups of variables whose ratios are already fixed, until one group remains.
  Merge-absorptivity is a decidable semantic property, and every merge-absorptive system admits complete finite-quotient elimination.
  Satisfiability is nevertheless \(\NEXP\)-hard even on this class.
  Finally, we exhibit a polynomial-size family for which a fixed sequence of quotient replacements has a unique surviving branch containing an exponentially long coefficient; every affine equation defining the same positive solution set also contains one.
\end{abstract}

\clearpage

\phantomsection\label{page:contents}
\tableofcontents
\clearpage
\pagenumbering{arabic}
\pagestyle{headings}

\section{Introduction}\label{sec:introduction}

Existential Presburger arithmetic is \(\NP\)-complete~\cite{Haase2018}.
Existential Presburger arithmetic with divisibility (EPAD) additionally allows atoms \(L\mid M\), where \(L\) and \(M\) are affine terms, that is, integer linear expressions \(c+\sum_i a_ix_i\).
Such an atom holds when \(M=qL\) for some integer \(q\).
Decidability was proved independently by Lipshitz~\cite{Lipshitz1978} and Bel'tyukov~\cite{Beltyukov1980}.
Lechner, Ouaknine, and Worrell proved an \(\NEXP\) upper bound~\cite{LechnerOuaknineWorrell2015}.
Previously, only \(\NP\)-hardness was known~\cite{Lipshitz1978}, and several works discussed the possibility of an \(\NP\) upper bound~\cite{BundalaOuaknine2014,HaaseThesis2012}.

EPAD occurs naturally in verification and string solving.  We mention two
salient occurrences.
A \emph{word equation} equates two patterns built from letters and word variables, and a \emph{length constraint} is a Presburger condition on the lengths of a solution.
Divisibility arises because commuting words are powers of one common word.
A \emph{parametric one-counter automaton} has one counter and transitions that test for zero or add a binary constant, a nonnegative integer parameter, or its negation.
Its reachability problem translates into EPAD and back~\cite{HaaseKreutzerOuaknineWorrell2009}.
Further uses include quantified synthesis problems~\cite{Lechner2015Synthesis,PerezRaha2022}, flat string-solving fragments~\cite{LinMajumdar2021}, parametric timed automata~\cite{BundalaOuaknine2014}, and logics for dynamic memory structures~\cite{BozgaIosif2005}.

Lipshitz-style decision procedures normalize an EPAD formula into finitely many divisibility systems whose local arithmetic can be analyzed.
This gives decidability and, with exponential witness bounds~\cite{LechnerOuaknineWorrell2015}, the known \(\NEXP\) upper bound.
The open question was whether the same route could be pushed down to \(\NP\).
Our main result rules that out; its proof is completed in \cref{sec:nexp-completeness}.

\begin{restatable}[Main theorem]{theorem}{EpadComplete}\label{thm:epad-complete}
EPAD satisfiability is \(\NEXP\)-complete under polynomial-time many-one reductions.
\end{restatable}

We further exhibit a computational limit of normalization itself: on a polynomial-size family, one elimination order produces an equation containing an exponential-bit coefficient, and every affine equation defining the same positive solution set also contains one (\cref{thm:replacement-blowup}).

\subparagraph*{Overview of the lower bound.}
The source problem is Succinct Exact-2-in-4-SAT: a polynomial-size Boolean circuit, on input a clause address, returns its four literals, each as a variable index and a sign, and a clause is satisfied when exactly two of these literal occurrences are true.
We prove its \(\NEXP\)-completeness in \cref{prop:succinct-exact-two-nexp}.

If the succinct formula has \(\ell\) variables, write \(\vec\alpha=(\alpha_0,\ldots,\alpha_{\ell-1})\) for their truth-value signs, with \(+1\) meaning true and \(-1\) meaning false.
For clause \(k\), let \(s_k(\vec\alpha)\) be the signed sum of its four literal occurrences: a positive occurrence of variable \(i\) contributes \(\alpha_i\), and a negative occurrence contributes \(-\alpha_i\).
Exactly two occurrences are true if and only if \(s_k(\vec\alpha)=0\).
Since \(s_k(\vec\alpha)/2\in\{-2,-1,0,1,2\}\), the packed equation \(\mathcal E(\vec\alpha)=\sum_k s_k(\vec\alpha)4^k=0\) holds exactly when every clause score is zero.
Since the succinct formula may contain exponentially many variables, we cannot introduce one EPAD variable for every sign.
A Chinese-remainder construction instead represents all signs by one signed integer witness subject to two divisibility constraints.

Some constants in this construction have exponentially many bits and therefore
cannot be written into the output formula.  We show instead that any requested
bit of the constants is computable in polynomial space from this succinct
formula description, and that \(\PSPACE\)-computable values are EPAD-expressible.  To
show this latter result, we rely on \emph{fixed-width arithmetic logic circuits}
(ALCs), and show that polynomial-time uniform polynomial-size ALCs compute
exactly the class \(\FPSPACE\) of nonnegative-integer functions computable in
polynomial space, and then provide an ALC-to-EPAD translation that gives
polynomial-size EPAD definitions of the ALC outputs.  We note, as an aside, that
allowing exponentially many ALC gates, provided each gate's type and predecessors can be recovered from its binary index in polynomial time, gives
a characterization of the exponential-time function class \(\FEXP\), which is
similarly included in EPAD-expressible functions.

\subparagraph*{Finite-quotient elimination.}
Lipshitz-style procedures~\cite{Lipshitz1978,LechnerOuaknineWorrell2015} replace a divisibility \(L\mid M\) by the finitely many quotient equations \(M=qL\) allowed by the current constraints, producing a search tree of quotient choices.
We call a system \MergeAbs{} when one sequence of such steps, fixed in advance and independent of the quotient choices, pins every variable to a fixed rational multiple of a single unknown, after which every remaining divisibility is a constant test.
This differs from the \(\NP\) fragment of D\'efossez et al.~\cite{DefossezHaaseMansuttiPerez2024}, which instead requires a variable order making divisibility dependencies one-way.
We prove that merge-absorptivity is decidable, while satisfiability is already \(\NEXP\)-hard even when restricted to merge-absorptive systems.
We also give a polynomial-size family with a fixed quotient-replacement schedule whose unique surviving branch requires an exponential-bit coefficient (\cref{thm:replacement-blowup}).

\subparagraph*{Consequences of the lower bound.}
We derive the following consequences of the main theorem, with the automata proof in \cref{sec:nexp-completeness} and the word-equation proof in Appendix~\ref{sec:word-equations}.
For word equations, \emph{nonerasing} means that no variable is assigned the empty word, and a \emph{positive Boolean combination} uses only conjunction and disjunction at the outer level.
\begin{restatable}{corollary}{DownstreamHardness}\label{cor:downstream-hardness}
Nonerasing satisfiability for positive Boolean combinations of word equations and quantifier-free Presburger length constraints over the fixed alphabet \(\{a,b\}\), and reachability for parametric one-counter automata, are \(\NEXP\)-hard.
The latter problem is in \(\NEXP\), and hence \(\NEXP\)-complete.
\end{restatable}

\noindent
Day and Konefal show that every EPAD-definable relation can be built from length abstractions of word equations~\cite{DayKonefal2025}.
The main theorem therefore gives the word-equation lower bound.
Its nonerasing two-letter specialization is spelled out for completeness in Appendix~\ref{sec:word-equations}.
The parametric one-counter reductions are due to Haase et al.~\cite{HaaseKreutzerOuaknineWorrell2009}.

\Needspace{10\baselineskip}
\subparagraph*{Relation to interactive proof systems.}
Like our source circuit, verifiers of exponentially long probabilistically
checkable proofs (PCPs) for \(\NEXP\) use polynomial-length
addresses~\cite[Proposition~10.8(4)]{BellareGoldreichSudan1998}.
Babai et al.'s multi-prover protocol over finite fields (to establish \(\MIP=\NEXP\))
tests a polynomial at a random field element to check that its coefficients
all vanish~\cite[Section~7.1]{BabaiFortnowLund1991}.
Our fixed weights \(4^k\) instead make one exact congruence equivalent
to \(s_k(\sigma,\vec\alpha)=0\) for every clause \(k\)
(\cref{lem:exact-one-congruence}).

\addtolength{\abovedisplayskip}{-3pt}
\addtolength{\belowdisplayskip}{-3pt}
\addtolength{\abovedisplayshortskip}{-3pt}
\addtolength{\belowdisplayshortskip}{-3pt}
\section{Background and Technical Overview}\label{sec:technical-overview}

This section states the main intermediate results and sketches their proofs; the full arguments follow in the subsequent sections.

We reduce from Succinct Exact-2-in-4-SAT\@.
A succinct instance \(\mathcal I\) consists of an exponentially large Exact-2-in-4 formula described by a polynomial-size Boolean circuit that returns any addressed clause.
The formal definition is in \cref{def:succinct-exact-two}.
Its \(\NEXP\)-completeness is proved in \cref{prop:succinct-exact-two-nexp}.

Given a succinct instance \(\mathcal I\), the reduction constructs in polynomial time a polynomial-size EPAD sentence \(\Xi_{\mathcal I}\) that is satisfiable exactly when the Exact-2-in-4 formula described by \(\mathcal I\) is satisfiable.

Throughout this section, variables range over \(\Z\), and constants and coefficients are binary encoded, so their bit lengths count toward formula size.

Our main technical contribution is the observation that divisibility can encode information into very large numbers.

\subsection{Divisibility as a Compressor}
\label{subsec:overview-divisibility}

The reduction needs short definitions of multipliers that are too long to write.
For every \(t\ge0\), the \emph{Fermat ladder} \(\Delta_t(x,r)\) below has \(O(t\log t)\) divisibility atoms and coefficients of polynomial bit length in \(t\), yet its solutions satisfy exactly
\begin{equation}\label{eqn:ladder-relation}
 r=(2^{2^t}-1)x.
\end{equation}
The forced multiplier has \(2^t\) bits and doubly exponential value in \(t\).
Polynomial-size systems with exponentially long solutions were already known~\cite{LechnerOuaknineWorrell2015}.
Here \(r/x=2^{2^t}-1\), and hence \(r+x=2^{2^t}x\).
This gives a short definition of multiplication by the fixed power \(2^{2^t}\), without writing that power as a coefficient.

The successive ratios used below are Fermat numbers of the form \(2^{2^q}+1\).
For \(t\ge0\), take variables
\begin{equation}\label{eqn:order-constraints}
0<y_0<\cdots<y_t,
\end{equation}
where \(y_0=x\) and \(y_t=r\).
At each stage \(i=t,\ldots,1\), impose the following four affine divisibility constraints:
\begin{equation}\label{eqn:fmat-atoms}
 y_{i-1}\mid y_i,
 \qquad
 y_i-y_{i-1}\mid y_t+y_i,
 \qquad
 y_i-2y_{i-1}\mid y_t,
 \qquad
 y_i+y_{i-1}\mid y_t+3y_i.
\end{equation}
The first three atoms restrict the stage ratio to finitely many Fermat-number candidates, and the exclusion atoms below select the intended one.
The fourth atom does not affect the forced ratio.
Its role is to permit elimination by finitely many quotient equations: once \(y_t/y_i\) is fixed, its quotient ranges over a finite set, and each feasible integer quotient determines \(y_{i-1}/y_i\).

At stage \(i\), put \(h=t-i\).
For \(h\ge1\), let \(m_h\) be the least nonnegative integer satisfying \(2^{m_h+1}/(m_h+2)>h\), and set \(D_h=\min\{m_h,h-1\}\).
For every \(0\le q\le D_h\), add the exclusion atom
\begin{equation}\label{eqn:exclusion}
 y_i-(2^{2^q}+1)y_{i-1}
 \mid y_t+(1-2^{2^q})y_i.
\end{equation}
No exclusion atom is needed when \(h=0\).

The system \(\Delta_t(x,r)\) consists of the order constraints from \cref{eqn:order-constraints}, the four atoms from \cref{eqn:fmat-atoms} at every stage, and the exclusion atoms from \cref{eqn:exclusion}.
Its construction realizes the identity
\[
  \prod_{q=0}^{t-1}\bigl(2^{2^q}+1\bigr)=2^{2^t}-1
\]
and therefore forces the advertised relation from \cref{eqn:ladder-relation}.
At \(t=2\), the two stages leave the quotients \(3\) and \(5\), so \(y_1=5x\) and \(r=y_2=15x=(2^{2^2}-1)x\).

\begin{sloppypar}
\proofparagraph{Proof idea}Correctness is an induction on \(h=t-i\), whose invariant is \(y_t=(2^{2^h}-1)y_i\).
Suppose \(y_t=(B-1)y_i\) with \(B=2^{2^h}\), and write \(y_i=Qy_{i-1}\).
The first three atoms force \(Q=2^{2^a}+1\) for some \(0\le a\le h\), so the stage ratio is a Fermat number.
The written exclusions rule out \(a\le D_h\) directly.
If \(D_h=h-1\), no smaller candidate remains.
Otherwise \(D_h=m_h\), and a surviving \(a<h\) would make every integer in \(a-m_h,\ldots,a\) divide \(h-a<h\).
The least common multiple of any \(d+1\) consecutive positive integers is at least \(2^{d+1}/(d+2)\) (\cref{lem:lcm-interval}); taking \(d=m_h\) makes it greater than \(h\), a contradiction.
Thus only \(O(\log h)\) written exclusions are needed, \(Q=B+1\), and \((B-1)(B+1)=B^2-1\) advances the invariant.
The full argument is in \cref{lem:mersenne-multiplier,rem:mersenne-exclusions}.\qed
\end{sloppypar}

Having obtained divisibility as a compression mechanism, we next need an arithmetic model capable of manipulating the resulting exponentially wide values so as to encode hard combinatorial computations.
This role is played by arithmetic logic circuits.

\subsection{Arithmetic Logic Circuits and Indexed Bits}
\label{subsec:overview-alc}

A circuit can describe a machine word too wide to write by specifying how its bits are produced.
For each input length \(n\), an arithmetic logic circuit (ALC) \(C_n\) has \(n\) Boolean inputs, constant gates \(0\) and \(1\), one distinguished output gate, and a common word-width \(W(n)\).
It computes a function \(C_n:\bits^n\to\{0,\ldots,2^{W(n)}-1\}\), whose output is interpreted as a nonnegative \(W(n)\)-bit integer.
Its gates compute addition and multiplication modulo \(2^{W(n)}\), together with logical left and right shifts by fixed nonnegative offsets.

We use polynomial-time \emph{local} uniformity for both circuit sizes considered below, with the gates of \(C_n\) numbered \(0,\ldots,S(n)-1\) in topological order.
On input \(1^n\), a uniformity algorithm returns \(S(n)\), \(W(n)\), and the output-gate index in binary.
Given also a gate index \(i<S(n)\) in binary, it returns the gate type, predecessor indices, and input index or shift offset when applicable.
Even when \(S(n)\le2^{\poly(n)}\), the index has polynomial length, so this interface gives local access without constructing the whole exponential circuit in polynomial time.
For polynomial-size circuits, it is equivalent to explicit polynomial-time construction.
The class \(\polyALC\) restricts \(S(n)\) to polynomial size, while \(\expALC\) permits \(S(n)\le2^{\poly(n)}\) with the same ALC model and local-uniformity condition.
In both cases \(1\le W(n)\le2^{\poly(n)}\).

We write \(\FPSPACE\) for nonnegative-integer function families on Boolean inputs computed by deterministic polynomial-space transducers with a one-way write-only output tape, and \(\FEXP\) for the analogous deterministic \(2^{\poly(n)}\)-time families.

For a function family \(f=(f_n)_n\), indexed-bit access supplies a predicate \(\mathcal B(\vec{x},i)\) that returns the \(i\)th bit of \(f_n(\vec{x})\).
It is equivalent to the transducer model above: \(f\) belongs to \(\FPSPACE\) if and only if there are a polynomial \(p\) and a predicate \(\mathcal B\in\PSPACE\) such that \(f_n(\vec{x})<2^{2^{p(n)}}\) and, for every binary index \(0\le i<2^{p(n)}\), \(\mathcal B(\vec{x},i)\) is the \(i\)th bit of \(f_n(\vec{x})\).
We use the standard indexed-bit characterization of \(\FPSPACE\)~\cite{Ladner1989}, which we recall and prove for completeness in \cref{lem:indexed-bits}.

For polynomial-size ALC families, polynomial space recovers one output bit by following shifts and scanning lower-order bits and their carries (\cref{prop:alc-to-fpspace}).
The choice of these four gate types is not restrictive: adding standard fixed-width operations such as subtraction, comparisons, division, rotations, and bitwise operations does not change the class of functions computed by polynomial-size ALCs (\cref{prop:alc-standard-operations}).

Polynomial-size families characterize \(\FPSPACE\).
Allowing exponentially many gates under the same local-uniformity convention gives the analogous class \(\FEXP\) (\cref{prop:expalc-fexp}).
Encoding such circuits into EPAD requires an additional layer of compression.

\subsection{Packing Truth Tables and FPSPACE}
\label{subsec:overview-packed-tables}

We therefore first treat the polynomial-size case, where the main ideas are easier to see: we show how ALCs encode polynomial-space computations and then compile those circuits into EPAD.
We sketch the converse inclusion \(\FPSPACE\subseteq\polyALC\) by packing truth tables into integers so that ALC operations evaluate Boolean connectives and quantifiers in parallel.

For Boolean functions \(F,G:\bits^k\to\bits\), one integer can hold an entire truth table in widely separated bit slots.
Packing Boolean data is related to arithmetizations of polynomial space~\cite{BabaiFortnow1991,LundFortnowKarloffNisan1992,Shamir1992}, but our identities are deterministic, with no verifier interaction and no random challenges.

For \(F\colon\bits^k\to\bits\), let \(\enc(\vec a)\) be the integer whose base-five digits are \(a_0,\ldots,a_{k-1}\), from least to most significant.
At scale \(R\), the packed integer
\[
  \langle F\rangle_R
  =\sum_{\vec a\in\bits^k}F(\vec a)2^{R\enc(\vec a)}
\]
is its \emph{packed truth table} at scale \(R\), which we shorten to packed table.

The tuple \(\vec a\) contains Boolean values, \(\enc(\vec a)\) is its base-five address, and \(R\enc(\vec a)\) is the binary position where the table entry is stored.

Multiplying \(\langle F\rangle_R\) and \(\langle G\rangle_R\) generates one term for every pair of assignments.
The operation \(\operatorname{block}_{h,d}(x)\) denotes the \(d\) consecutive bits of \(x\) starting at position \(h\).
For one-variable Boolean functions \(F,G\), let \(\delta\) be a nonnegative integer giving the desired separation between the two retained entries.
Before alignment, the product \(\langle F\rangle_R \langle G\rangle_R\) is \(F(0)G(0)+(F(0)G(1)+F(1)G(0))2^R+F(1)G(1)2^{2R}\).
The first and last terms are the desired diagonal products, with equal input bits.
The two middle products are cross terms.
Choose \(R>\delta+4\).
Then
\[
 \operatorname{block}_{2R,R}\!\left(
   \langle F\rangle_R\langle G\rangle_R(2^{2R}+2^\delta)
 \right)
 =F(0)G(0)+F(1)G(1)2^\delta.
\]
The two diagonal products land at positions \(2R\) and \(2R+\delta\).
The cross terms either remain below position \(2R\) without carrying into the block or begin at position \(3R\).

For output offsets \(\vec\delta=(\delta_0,\ldots,\delta_{k-1})\), define
\[
 \mathcal M_{R,\vec\delta}(F,G)
 =\langle F\rangle_R\langle G\rangle_R
   \prod_{i=0}^{k-1}\bigl(2^{2R5^i}+2^{\delta_i}\bigr).
\]
An expanded term is determined by \(\vec a,\vec b,\vec z\in\bits^k\), where \(\vec a\) and \(\vec b\) select entries from the two packed tables and \(z_i=1\) indicates the choice of \(2^{2R5^i}\) from the \(i\)th alignment factor.
Its base-five digit is \(a_i+b_i+2z_i\in\{0,\ldots,4\}\), so no carries occur between base-five positions, and every digit is two exactly when \(\vec a=\vec b\) and \(z_i=1-a_i\) for every \(i\).
For a desired output scale \(R'\), take \(\delta_i=R'5^i\) and \(R>\sum_i\delta_i+3k+1\).  The block of \(R\) bits at position \(R(5^k-1)/2\) is then the packed table \(\langle F\wedge G\rangle_{R'}\) (\cref{lem:diagonal-extraction,lem:pointwise}).

To evaluate Boolean formulas, maintain a pair of packed tables for each predicate: one for its true assignments and one for its false assignments.
Pointwise product implements conjunction, exchanging the paired tables implements negation, and disjunction follows from conjunction and complement.
Restricting a quantified variable to \(0\) and to \(1\) gives two tables on the remaining variables, and combining them by the same operations implements \(\forall y\) and \(\exists y\) (\cref{lem:restriction}).
The paired-table invariant therefore persists through nested quantifiers.

Fix a function family \(f=(f_n)_n\in\FPSPACE\), and let \(p\) and \(\mathcal B(\vec x,i)\) be the polynomial and indexed-bit predicate from the characterization above.
Let \(\vec a=(a_0,\ldots,a_{p(n)-1})\in\bits^{p(n)}\) be the binary representation, least significant bit first, of an output-bit index \(i<2^{p(n)}\), so that \(i=\sum_{j=0}^{p(n)-1}a_j2^j\).
Thus \(i\) is a bit position in the exponentially long output \(f_n(\vec x)\), whereas \(\enc(\vec a)=\sum_j a_j5^j\) is the base-five address used to store \(\mathcal B(\vec x,i)\) in the packed table.
A uniform reduction from \(\PSPACE\) to quantified Boolean formulas (QBF)~\cite[Chapter~19]{Papadimitriou1994} gives a polynomial-size prenex formula
\[
 \Psi_n(\vec x,\vec a)=Q_0y_0\cdots Q_{q-1}y_{q-1}\,\varphi_n(\vec x,\vec a,\vec y),
 \qquad \vec y=(y_0,\ldots,y_{q-1}),
\]
for some \(q=\poly(n)\), with \(\varphi_n\) quantifier-free.
The QBF is true exactly when \(\mathcal B(\vec x,i)=1\).
The circuit packs \(\vec a,y_0,\ldots,y_{q-1}\), while \(\vec x\) remains its Boolean input.
Over the \(k\) variables still indexing the table, the all-ones table at scale \(R\) is \(\prod_{i<k}(1+2^{R5^i})\), and omitting or shifting one factor gives the false and true tables of a packed variable, each using \(O(k)\) gates.
Choose the common ALC word-width \(W\) large enough for the packed computation.
An external input bit \(x_t\) then needs its complement.
Writing \(+_W\) and \(\ll_W\) for addition and left shift modulo \(2^W\), and \(\gg\) for right shift, the circuit forms it as \(\bar x_t=((x_t+_W1)\ll_W(W-1))\gg(W-1)=1-x_t\).
Evaluating \(\varphi_n\) and eliminating quantified variables from the inside out leaves a table indexed by \(\vec a\) whose entry at address \(\vec a\) is \(\mathcal B(\vec x,i)\).
Those entries are Boolean, so choosing \(\delta_j=2^j\) places each one in output bit position \(i\), producing \(f_n(\vec x)\).
The construction has polynomially many gates, and every scale, offset, and the exponential word width has a polynomial-length binary description.
The full proofs of the two inclusions are given in \cref{prop:alc-to-fpspace,thm:fpspace-to-alc}, establishing the following characterization.

\begin{restatable}[Arithmetic-circuit characterization]
  {theorem}{AlcFpspace}\label{thm:alc-fpspace}
Polynomial-time uniform polynomial-size ALC families compute exactly \(\FPSPACE\): \(\polyALC=\FPSPACE\).
\end{restatable}

\subsection{Functional EPAD Compilation}
\label{subsec:overview-gates}

We now compile polynomial-size ALC families into polynomial-size existential EPAD formulas whose unique outputs are the corresponding circuit outputs.
Combined with the previous subsection, this gives short EPAD definitions of the polynomial-space-computable constants needed by the reduction.
Our compiler first constructs the powers of two determined by the word width and shift offsets, then translates each gate operation.

An existential EPAD formula \(\Phi_n(\vec x,y)\) defines \(f_n:\bits^n\to\N\) functionally if, for every \(\vec a\in\bits^n\) and \(b\in\Z\), one has \(\Phi_n(\vec a,b)\Longleftrightarrow b=f_n(\vec a)\).
The class \(\FEPAD\) consists of families with polynomial-time uniform, polynomial-size functional formulas.

\proofparagraph{Large powers of two}For every fixed nonnegative integer \(e\) given in binary, \(\Scale_e(A,B)\) is an EPAD formula of size polynomial in the binary length of \(e\) that holds exactly when \(B=2^eA\).
Writing \(2^e\) as a coefficient would require \(e+1\) bits and make the EPAD formula exponentially large.
Instead, let \(J=\{j:e_j=1\}\) be the set of nonzero binary positions of \(e\), so \(e=\sum_{j\in J}2^j\).
On a positive current value \(z\), the Fermat ladder \(\Delta_j(z,r)\) from \cref{subsec:overview-divisibility} forces \(r=(2^{2^j}-1)z\).
Following it with \(z'=r+z\) therefore multiplies \(z\) by \(2^{2^j}\).
Chaining the set bits \(j\in J\) multiplies by \(2^e\), and \(e=0\) uses equality.
For signed inputs, decompose \(A\) as a difference of positive values, scale both values, and subtract the results.

For an ALC of word-width \(W\), the constraints \(\Scale_W(1,P)\) and \(\Scale_{2W}(1,P_2)\) introduce the modulus \(P=2^W\) and the product bound \(P_2=2^{2W}=P^2\).
For every gate \(g\), let \(Y_g\) represent its value and impose \(0\le Y_g<P\).
For an input gate corresponding to Boolean input \(x_i\), set \(Y_g=x_i\).
Set constant gates to zero or one.
The remaining gate types are compiled as follows.

\proofparagraph{Addition}For \(g=u+_W v\), impose \(P\mid Y_u+Y_v-Y_g\).
The bound \(0\le Y_g<P\) forces \(Y_g\) to be the unique remainder of \(Y_u+Y_v\) modulo \(P\).

\proofparagraph{Multiplication}For \(g=u\times_W v\), the ordinary integer product \(T=Y_uY_v\) before reduction modulo \(P\) lies in \([0,P_2)\).
Introduce \(T\) with this bound.
Since EPAD terms are affine, the formula must force the equality indirectly.
For the correctness argument, let \(\Lambda=2^{2W+2}=4P_2\).
The formula never writes \(\Lambda\) as a coefficient.
Instead, it introduces \(Z_v,Z_T\) and imposes
\begin{equation}\label{eq:overview-product-atoms}
 \Scale_{2W+2}(Y_v,Z_v),
 \qquad
 \Scale_{2W+2}(T,Z_T),
 \qquad
 Y_v\mid T,
 \qquad
 1+Z_v\mid Y_u+Z_T.
\end{equation}
The scaling constraints mean \(Z_v=\Lambda Y_v\) and \(Z_T=\Lambda T\).
If \(Y_v\ne0\), let \(q=T/Y_v\), the quotient in the first divisibility.
Since \(T=qY_v\), scaling also gives \(Z_T=qZ_v\).
Modulo \(1+Z_v\), one has \(Z_v\equiv-1\), and hence \(Y_u+Z_T\equiv Y_u-q\).
The second divisibility therefore gives \(1+\Lambda Y_v\mid Y_u-q\).
Because \(q<P_2\), \(\lvert Y_u-q\rvert<P+P_2<\Lambda\).
The divisor is at least \(1+\Lambda\), so \(q=Y_u\).
If \(Y_v=0\), then \(Y_v\mid T\) is \(0\mid T\), which under ordinary integer divisibility holds exactly when \(T=0\).
Conversely, if \(T=Y_uY_v\), both divisibilities admit the same quotient \(Y_u\).
Thus, when \(Y_v\ne0\), the two divisibility constraints force the same quotient \(Y_u\), with \(T/Y_v=Y_u\).
We call this quotient sharing.
Thus the constraints force the ordinary product, and \(P\mid T-Y_g\) then selects its unique residue in \([0,P)\).

\proofparagraph{Left shift}For \(g=u\ll_W s\), set \(Y_g=0\) if \(s\ge W\).
Otherwise, introduce a fresh \(T\).
The constraint \(\Scale_s(Y_u,T)\) forces the ordinary value \(T=2^sY_u\), and \(P\mid T-Y_g\) selects its unique residue in \([0,P)\).

\proofparagraph{Right shift}For \(g=u\gg s\), set \(Y_g=0\) if \(s\ge W\).
Otherwise, introduce a divisor \(S\), a multiple \(Q\), and a remainder \(R\), and impose \(\Scale_s(1,S)\), \(\Scale_s(Y_g,Q)\), \(Y_u=Q+R\), and \(0\le R<S\).
The scaling constraints give \(S=2^s\) and \(Q=2^sY_g\).
Thus \(Y_u=Q+R\) is quotient--remainder division by \(S\), with quotient \(Y_g\) and remainder \(R\).
It forces \(Y_g=\lfloor Y_u/2^s\rfloor\).

In topological order, every gate relation uniquely determines its bounded output from the values of its predecessors.
A family in \(\polyALC\) has polynomially many gates, and every scaling exponent has polynomial binary length.
Each gate therefore contributes a polynomial-size formula.
Complete proofs and bounds appear in \cref{sec:divisibility-compressor,sec:poly-alc-fepad}.
Together with \cref{thm:alc-fpspace}, this gives \(\FPSPACE\subseteq\FEPAD\), which supplies the short definitions used next for the reduction's exponentially long constants.

\subsection{From Succinct Exact-2-in-4 to EPAD}
\label{subsec:overview-reduction}

The reduction assembles the three preceding subsections into one polynomial-size EPAD sentence.
Let \(\mathcal I=(D,m,\ell)\) be a succinct Exact-2-in-4 instance.

\begin{overviewdef}[Succinct Exact-2-in-4 instance, formally \cref{def:succinct-exact-two}]
The integers \(m,\ell\ge1\), giving the numbers of clauses and variables, are written in binary, and a polynomial-size Boolean circuit \(D\), given the binary representation of a clause index \(0\le k<m\), returns the four literals of clause \(k\), each as a variable index and a sign.
The described formula is denoted \(F_{\mathcal I}\); each clause requires exactly two true literal occurrences.
\end{overviewdef}

The decision problem asks whether \(F_{\mathcal I}\) is satisfiable; we prove its \(\NEXP\)-completeness in \cref{prop:succinct-exact-two-nexp}.

We prove the following normal form in \cref{sec:succinct-exact-two-normal-form}. The remainder of this subsection outlines its construction and use in the reduction.

\begin{restatable}[Succinct Exact-2-in-4 normal form]
  {theorem}{OneWitnessNormalForm}\label{thm:succinct-exact-two-normal-form}
For every succinct Exact-2-in-4 instance \(\mathcal I=(D,m,\ell)\), there are positive integers \(M,H,K\), all determined by \(\mathcal I\), such that \(F_{\mathcal I}\) is satisfiable if and only if there exists \(X\in\Z\) satisfying
\[
 -H\le X\le H,\qquad
 M\mid X,\qquad
 K\mid H^2-X^2.
\]
\end{restatable}

The modulus \(M\) checks all clauses directly on \(X\), while the square-difference condition involving \(H\) and \(K\) ensures that \(X\) encodes a choice of signs.

\proofparagraph{Zero-target clause packing}For \(\vec u\in\bits^\ell\), put \(\alpha_i=2u_i-1\in\{-1,+1\}\).
If clause \(k\) contains literals \((i_h,\varepsilon_h)\), for \(h=1,2,3,4\), where \(i_h\) is the variable index and \(\varepsilon_h\in\{-1,+1\}\) is its polarity (\(+1\) for a positive literal and \(-1\) for a negated literal), put
\[
 s_k(\vec\alpha)=\sum_{h=1}^4\varepsilon_h\alpha_{i_h}.
\]
The equality \(s_k=0\) says that clause \(k\) has exactly two true literal occurrences.
With
\[
 \mathcal E(\vec\alpha)=\sum_{k=0}^{m-1}s_k(\vec\alpha)4^k,
 \qquad M=2\cdot4^m,
\]
the base-four spacing gives
\[
 \vec u\models F_{\mathcal I}
 \quad\Longleftrightarrow\quad
 \mathcal E(\vec\alpha)\equiv0\pmod M.
\]

\proofparagraph{The constants \(H\) and \(K\)}Write
\[
 \mathcal E(\vec\alpha)=\sum_{j=0}^{\ell-1}c_j\alpha_j.
\]
Manders and Adleman encode a choice of signs by distributing pairwise coprime prime powers between \(H-X\) and \(H+X\), the two factors of a square difference~\cite[Lemma~1]{MandersAdleman1978}.
Using the first \(\ell\) odd primes, let \(K\) be the product of suitably large powers of these primes.
The construction chooses positive weights \(\theta_j\) with \(\theta_j\equiv c_j\pmod M\).
Each \(\theta_j\) is divisible by every prime-power factor of \(K\) except the one associated with coordinate \(j\), whose underlying prime does not divide \(\theta_j\).
Set \(H=\sum_j\theta_j\), choosing the prime powers large enough to ensure \(2H<K\).
Then
\[
 X=\sum_j\alpha_j\theta_j
\]
satisfies \(|X|\le H\) and \(K\mid H^2-X^2\), and conversely those two conditions recover a unique sign vector.
Since \(X\equiv\mathcal E(\vec\alpha)\pmod M\), the remaining condition is simply \(M\mid X\).

\proofparagraph{Why ALCs enter}The constants \(H(\mathcal I)\) and \(K(\mathcal I)\) have exponentially many bits. In \cref{prop:hk-fpspace}, we prove that their indexed bits are computable in polynomial space by querying the clause-description circuit \(D\) only at addressed clauses and performing the required arithmetic by indexed access.
Hence the maps \(\mathcal I\mapsto H(\mathcal I)\) and \(\mathcal I\mapsto K(\mathcal I)\) are \(\FPSPACE\) functions.
The characterization \(\polyALC=\FPSPACE\) converts these procedures into polynomial-size ALC families, and the functional ALC-to-EPAD translation gives polynomial-size EPAD formulas for their values.
For a concrete source instance, the bits of \(\mathcal I\) are fixed to their constants in those formulas.
This is the only use of ALCs in the \(\NEXP\)-hardness reduction.

\proofparagraph{The remaining arithmetic}The final EPAD sentence keeps a signed variable \(X\), imposes \(-H\le X\le H\) and \(M\mid X\), and uses the quotient-sharing multiplication gadget from the previous subsection to define
\[
 Y=(H-X)(H+X).
\]
The last test is \(K\mid Y\).
Once \(H\) and \(K\) have been defined, this is the only additional use of the bounded-product gadget.
The resulting polynomial-size EPAD sentence \(\Xi_{\mathcal I}\) is satisfiable exactly when \(F_{\mathcal I}\) is satisfiable.
Section~\ref{sec:succinct-exact-two-normal-form} proves this one-integer encoding and shows that the exponentially large constants \(H\) and \(K\) can be represented by polynomial-size EPAD formulas, yielding the polynomial-time reduction \(\mathcal I\mapsto\Xi_{\mathcal I}\).
Together with the known \(\NEXP\) upper bound, Section~\ref{sec:nexp-completeness} then gives \(\NEXP\)-completeness.

Separately, Section~\ref{sec:expalc-fepad-section} extends the same compression ideas to locally uniform exponential-size ALCs, yielding polynomial-size functional EPAD definitions for \(\FEXP\) functions.

\subsection{Finite-Quotient Elimination}
\label{subsec:overview-fragment}

Successive quotient equations can force all variables to be fixed rational multiples of a single variable.
Merge-absorptivity captures this behavior.
It is a semantic property; in \cref{prop:merge-absorptive-recognition}, we prove that it is decidable and that every merge-absorptive system admits complete finite-quotient elimination.
Satisfiability remains \(\NEXP\)-hard even on this fragment.

\begin{overviewdef}[Finite-quotient replacement, after
Lipshitz~\cite{Lipshitz1978,LechnerOuaknineWorrell2015}, formally
\cref{def:finite-quotient-replacement}]
Suppose the current constraints imply that every valuation satisfying \(L\mid M\) with \(L\ne0\) has \(M/L\) in a finite nonempty set \(\mathcal Q\subseteq\Z\).
Replace the atom by one equation \(M=qL\) for each \(q\in\mathcal Q\).
Together these alternatives contain every solution with \(L\ne0\).
A solution with \(L=M=0\) satisfies every alternative because each equation becomes \(0=q\cdot0\).
\end{overviewdef}

Under \(0<x\le y\le3x\), replacing \(x\mid y\) gives the three children \(y=x\), \(y=2x\), and \(y=3x\).
For example, the middle child is written as \(y-2x=0\), whose variable coefficients are \(1\) and \(-2\).
In general, rewrite a child equation \(M=qL\) as \(M-qL=0\).
If every coefficient, including the constant term, is zero, discard the equation as a tautology.
If all variable coefficients are zero but the constant term is nonzero, discard the child as inconsistent.
Otherwise make the equation \emph{primitive} by dividing all coefficients by their positive greatest common divisor.
This preserves its integer solutions.

\proofparagraph{Positive homogeneous systems}A linear form is \emph{homogeneous} if it has no constant term.
A \emph{positive homogeneous divisibility system} is a conjunction of divisibility atoms between homogeneous integer linear forms and homogeneous equality and inequality \emph{side conditions}, over variables in \(\Npos\).
Scaling a valuation by a positive integer therefore preserves every constraint.

\proofparagraph{Components}After some quotient equations have been chosen, a \emph{component} is a set of variables that the side conditions and those equations force to be positive rational multiples of one common scale.
Initially every variable forms a singleton component.

A divisibility atom \(L\mid M\) is a \emph{merge atom} for two components \(A,B\) if, under the current side conditions and earlier quotient equations, it has only finitely many possible integer quotients, and every quotient equation \(M=qL\) that is consistent with those constraints fixes the relative scale of \(A\) and \(B\).
The two components can then be replaced by \(A\cup B\).

Start with the singleton components \(\{a\},\{b\},\{c\}\).
Under \(a>b>0\), the quotient in \(2a+b\mid5a\) is forced to be \(2\), and the equation \(5a=2(2a+b)\) forces \(a=2b\).
With \(a=2b\) and \(b>c>0\), the quotient in \(a\mid b+3c\) is forced to be \(1\), giving \(b=3c\).
The component evolution is
\[
 \{a\},\{b\},\{c\}
 \longrightarrow\{a,b\},\{c\}
 \longrightarrow\{a,b,c\}.
\]

\begin{overviewdef}[Merge-absorptive systems, formally \cref{def:merge-absorptive}]
A positive homogeneous divisibility system with \(r\) variables is \MergeAbs{} if one can fix in advance a sequence of \(r-1\) distinct divisibility atoms and the pair of current components merged at each stage such that, for every consistent sequence of earlier quotient choices, the next atom merges its prescribed components.
These \(r-1\) merges reduce the singleton partition to a single component.
\end{overviewdef}

We prove the following hardness result in \cref{subsec:merge-absorptive}. The decidability and hardness arguments are outlined below.

\begin{restatable}[NEXP-hardness on merge-absorptive systems]
  {corollary}{MergeAbsHard}\label{cor:mergeabs-hardness}
Satisfiability is \(\NEXP\)-hard even when restricted to \MergeAbs{} positive homogeneous divisibility systems.
\end{restatable}

\proofparagraph{Decidability}Once the ratios within the current components are fixed, a candidate merge depends only on the ratio \(z=t_B/t_A\) between two component scales.
Feasible values of \(z\) form a rational interval computable by linear programming, and the divisibility quotient is a linear-fractional function of \(z\).
One can therefore decide whether only finitely many integer quotients occur and whether every consistent quotient equation fixes \(z\).
Proposition~\ref{prop:merge-absorptive-recognition} makes this test effective.

\proofparagraph{Hardness}For the hardness direction, we first use the standard normalization of EPAD formulas to conjunctions of affine divisibility constraints over integer variables~\cite[Section~IV.A]{LechnerOuaknineWorrell2015}.
Proposition~\ref{prop:epad-affdivz} recalls this reduction in detail for completeness.
Replacing each integer variable by a difference of two positive variables gives a system \(\mathcal J\) over positive variables.
We then transform the resulting system into a positive homogeneous one.
The construction is arranged so that its quotient atoms successively merge all variables into a single component.

The small-solution bound supplies a bit-length bound \(\nu=2^{\poly(|\mathcal J|)}\) such that a satisfiable system \(\mathcal J\), with variables \(x_1,\ldots,x_k\), has a solution \(0<x_j<2^\nu\) for every \(1\le j\le k\).
Put \(E=\nu+\lceil\log_2(k+1)\rceil+1\), so that \(k2^\nu<2^E\).
Introduce a common positive scale \(u\) and a strictly decreasing chain
\[
 s_0>s_1>\cdots>s_k>0,
 \qquad u\mid s_{j-1}-s_j\quad(1\le j\le k),
\]
and force \(s_0=2^Eu\) by chaining the Fermat-ladder gadgets \(\Delta_h\) for the set bits of \(E\), as in \cref{subsec:overview-gates}.
Recall that \(\Delta_h(x,r)\) enforces \(r=(2^{2^h}-1)x\), so adding \(r+x\) realizes multiplication by \(2^{2^h}\).

The quotient represents the source variable \(x_j\), and every affine form is represented homogeneously according to
\[
 x_j=\frac{s_{j-1}-s_j}{u},
 \qquad
 L=c+\sum_{j=1}^k a_jx_j
 \quad\longmapsto\quad
 \widehat L=cu+\sum_{j=1}^k a_j(s_{j-1}-s_j).
\]
Under the quotient encoding, \(\widehat L=uL(x_1,\ldots,x_k)\).
Since \(u>0\), replacing every source form by its hatted form preserves equalities, inequalities, and divisibilities, including a zero divisor.
The target conjunction is therefore a positive homogeneous divisibility system.

The Fermat ladders used to enforce \(s_0=2^Eu\) first merge all variables of the scaling gadget into one component.
The atoms \(u\mid s_{j-1}-s_j\) then merge \(s_1,\ldots,s_k\) into that component one by one, because each chosen quotient fixes \(s_j\) relative to the common scale.
After these merges, the source divisibilities reduce to constant tests (\cref{prop:merge-absorptive-completeness}).
The translation is polynomial-time and preserves satisfiability, so restricting to positive homogeneous merge-absorptive systems loses no generality for this decision problem.

\subparagraph*{Large coefficients.}
The same Fermat-ladder gadget \(\Delta_j(x,r)\) from \cref{subsec:overview-divisibility}, which enforces \(r=(2^{2^j}-1)x\), also witnesses the coefficient-growth lower bound.
Writing \(M_j=2^{2^j}-1\), its projection onto the variables \(x,r\) is \(\{t(1,M_j):t\in\Npos\}\), which we call its \emph{projected ray}.

We prove the following coefficient bounds in \cref{subsec:large-coefficients}.

\begin{restatable}[Large coefficients in finite-quotient replacements]
  {theorem}{ReplacementBlowup}\label{thm:replacement-blowup}
The polynomial-size merge-absorptive family \((\Delta_j)\) admits a specified replacement sequence whose unique surviving branch contains an equation with a coefficient of bit length \(\Theta(2^j)\).
Every affine equation over \(x,r\) that defines the same projected ray with \(x,r>0\) contains a coefficient with at least \(2^j\) bits.
\end{restatable}

\proofparagraph{Proof sketch}Replacing the atoms \(y_i+y_{i-1}\mid y_j+3y_i\) in the order \(i=j,j-1,\ldots,1\) successively fixes one more variable relative to \(y_j\), so \(\Delta_j\) is merge-absorptive.
For \(j\ge1\), the final surviving child is \(y_j+3y_1=q(y_1+x)\), where \(q=2^{2^{j-1}}+1\), so it contains a coefficient of bit length \(\Theta(2^j)\).
Any equation \(ax+br+c=0\) defining the projected ray is nonzero, and substituting \((x,r)=(t,M_jt)\) for all \(t\in\Npos\) gives \(c=0\) and \(a=-bM_j\).
Hence \(b\ne0\), so \(|a|=|b|M_j\ge M_j\) and \(a\) has at least \(2^j\) bits.\qed

\subparagraph*{Discussion.}
The complexity of the Bozga--Iosif--Lechner fragment, the quantified extension of EPAD used in verification~\cite{BozgaIosif2005,LechnerThesis2016} and called BIL in~\cite{PerezRaha2022}, is open, as are satisfiability for word equations with length constraints and reachability for other parametric automata models such as parametric timed automata~\cite{BundalaOuaknine2014}.
We have implemented the normalization procedures used here, and the decision procedure of D\'efossez et al.~\cite{DefossezHaaseMansuttiPerez2024} is the natural next target.

\phantomsection\label{sec:technical-overview-end}
\clearpage
\addtolength{\abovedisplayskip}{3pt}
\addtolength{\belowdisplayskip}{3pt}
\addtolength{\abovedisplayshortskip}{3pt}
\addtolength{\belowdisplayshortskip}{3pt}
\section{Preliminaries and the Indexed-Bit Characterization}
\label{sec:preliminaries}

\subparagraph*{Conventions.}
The expression \(2^{\poly(n)}\) denotes \(2^{p(n)}\) for some polynomially bounded integer-valued function \(p\).
For \(a>0\), let \(\len(a)=\lfloor\log_2 a\rfloor+1\), and put \(\len(0)=1\).
Integers and affine coefficients are binary encoded, and their bit lengths count toward formula size.
Shifts are logical on nonnegative binary words, with bit zero least significant.

An \emph{affine-linear integer term} is \(c+\sum_i a_i x_i\).
It is \emph{homogeneous} when \(c=0\).

\begin{definition}[EPAD formula satisfiability]\label{def:epad-sat}
An EPAD formula is an existential first-order formula over integer variables using Boolean combinations of affine-linear equalities, affine-linear inequalities, and divisibility atoms \(L\mid M\) between affine-linear terms.
EPAD satisfiability asks whether a given such formula is satisfiable.
\end{definition}

Divisibility has its ordinary integer meaning: \(L\mid M\) if and only if \(M=qL\) for some \(q\in\Z\).
Thus \(0\mid0\) is true and \(0\mid M\) is false for \(M\ne0\).
When we invoke a result stated under the convention that the divisor must be nonzero, we separate the case \(L=M=0\) by a Presburger disjunction, without further comment.

\begin{definition}[Functional EPAD]\label{def:fepad}
An \(n\)-input functional EPAD formula is an EPAD formula \(\Phi_n(\vec x,y)\equiv\exists\vec z\,\Theta_n(\vec x,y,\vec z)\), where \(\vec x\in\bits^n\) are Boolean free inputs and \(y\) is the distinguished integer output.
It computes \(f_n\colon\bits^n\to\N\) if, for every \(\vec a\in\bits^n\) and every \(b\in\Z\),
\[
 \Phi_n(\vec a,b)\quad\Longleftrightarrow\quad b=f_n(\vec a).
\]
Only the output is required to be unique.
The auxiliary witnesses need not be.
The class \(\FEPAD\) consists of functions computed by polynomial-time uniform polynomial-size families of such formulas.
Zero-input families are allowed.
\end{definition}

\begin{definition}[Arithmetic logic circuits]\label{def:alc}
An \emph{arithmetic logic circuit} (ALC) of \emph{word-width \(W\)} is an acyclic circuit with a distinguished output gate whose values lie in \(\{0,\ldots,2^W-1\}\).
It has Boolean inputs, constants \(0,1\), and gates
\[
 \begin{aligned}
 u+_W v&=(u+v)\bmod 2^W,&
 u\times_W v&=uv\bmod 2^W,\\
 u\ll_W s&=2^s u\bmod 2^W,&
 u\gg s&=\left\lfloor u/2^s\right\rfloor.
 \end{aligned}
\]
Shift offsets are nonnegative binary parameters in the circuit description.
Shifts by at least \(W\) return zero.

An ALC family \((C_n)_n\) has a size \(S(n)\), with gates indexed \(0,\ldots,S(n)-1\) in topological order, and a word-width \(W(n)\) satisfying \(1\le W(n)\le2^{\poly(n)}\).
The family is \emph{polynomial-time uniform} if a polynomial-time algorithm, on input \(1^n\), returns the binary representations of \(S(n)\), \(W(n)\), and the output-gate index, and, on input \((1^n,i)\) with \(i<S(n)\) given in binary, returns the local description of gate \(i\): its type, predecessor indices, and Boolean input index or shift offset when applicable.
Predecessor indices are smaller than \(i\).

The class \(\polyALC\) consists of the functions computed by polynomial-time uniform families with \(S(n)\le\poly(n)\).
The class \(\expALC\) is defined identically except that \(S(n)\le2^{\poly(n)}\).
In both classes the output word is interpreted as a nonnegative binary integer.
\end{definition}

For polynomial-size families, this local definition of uniformity is equivalent to generating the whole circuit in polynomial time.
Its advantage is that exactly the same definition continues to apply when the circuit has exponentially many gates.

\begin{definition}[Function classes]\label{def:function-classes}
The class \(\FPSPACE\) consists of functions from Boolean strings to nonnegative integers computed by deterministic polynomial-space transducers with a one-way write-only output tape.
The class \(\FEXP\) is defined in the same way using deterministic \(2^{\poly(n)}\)-time transducers.
Outputs are read in binary, and leading zeros are immaterial.
The class \(\NEXP\) consists of languages accepted by nondeterministic machines in time \(2^{\poly(n)}\).
\end{definition}

\begin{lemma}[Indexed-bit characterization]\label{lem:indexed-bits}
A function family \(f=(f_n)_n\), with \(f_n\colon\bits^n\to\N\), belongs to \(\FPSPACE\) if and only if there are a polynomial \(p\) and a \(\PSPACE\) predicate \(\mathcal B\) such that \(f_n(\vec x)<2^{2^{p(n)}}\) and, for every binary index \(0\le i<2^{p(n)}\), \(\mathcal B(\vec x,i)=1\) exactly when the \(i\)th binary bit of \(f_n(\vec x)\) is \(1\).
\end{lemma}

\begin{proof}
Consider the standard write-only-output model for \(\FPSPACE\)~\cite{Ladner1989}.
Apart from its output, a halting polynomial-space transducer has only exponentially many configurations.
Repeating one would make it loop, so its output has exponential length.
Replaying the computation first determines that length and then recovers any requested least-significant-bit index in polynomial space.
Conversely, evaluate \(\mathcal B(\vec x,i)\) from the largest index down to zero and write each returned bit.
The index counter and each evaluation use polynomial space, and leading zeros are immaterial.
\end{proof}

\section{Divisibility as a Compressor}
\label{sec:divisibility-compressor}

The constructions in this section define, with polynomially many constraints, multiplication by a power of two whose exponent is given in binary, and the bounded product of two variables.
They are the arithmetic core of the lower bound, and \cref{sec:poly-alc-fepad} uses them to translate circuit gates.

If a formula has free variables \(x,r\), its \emph{\((x,r)\)-projection} is the set of pairs that extend to values of its auxiliary existential variables satisfying the formula.
The power-of-two constructions in this section use linear equalities, divisibility atoms, and positivity and order constraints.

A \emph{homogeneous divisibility system} consists of divisibility atoms between homogeneous integer linear forms and homogeneous linear side conditions, that is, equalities and inequalities.
It is \emph{positive} when all variables range over positive integers.

\subsection{Number-Theoretic Bounds Used by the Power-of-Two Gadgets}
\label{subsec:number-theoretic-gadgets}

The next two lemmas are the number-theoretic base of the Fermat ladder.
The first turns divisibility between Mersenne numbers into divisibility between their exponents.
The second bounds how many exclusion atoms a stage needs.

\begin{lemma}[Mersenne divisibility]\label{lem:mersenne-div}
For positive integers \(u,v\), one has \(2^u-1\mid2^v-1\) if and only if \(u\mid v\).
\end{lemma}

\begin{proof}
If \(v=ku\), factor \((2^u)^k-1\).
Conversely, write \(v=qu+r\) with \(0\le r<u\).
Modulo \(2^u-1\), one has \(2^v-1\equiv2^r-1\).
Divisibility then forces \(r=0\), since \(0\le2^r-1<2^u-1\).
\end{proof}

\begin{lemma}[LCM of a short interval]\label{lem:lcm-interval}
The least common multiple of any \(d+1\) consecutive positive integers is at least \(2^{d+1}/(d+2)\).
\end{lemma}

\begin{proof}
Put \(n=d+1\).
Among \(n\) consecutive integers, every residue class modulo a prime power \(p^e\le n\) occurs at least once.
In particular, the interval contains a multiple of \(p^e\).
Its lcm is therefore divisible by \(\lcm(1,\ldots,n)\).
For each \(0\le k\le n\), the binomial coefficient \(\binom nk\) divides this lcm.
Indeed, by Legendre's formula, the exponent of a prime \(p\) in \(\binom nk\) is
\[
 \sum_{e\ge1}\left(
   \left\lfloor\frac n{p^e}\right\rfloor
  -\left\lfloor\frac k{p^e}\right\rfloor
  -\left\lfloor\frac{n-k}{p^e}\right\rfloor
 \right).
\]
Each summand is zero or one, and there are at most \(\lfloor\log_p n\rfloor\) nonzero terms.
Since the \(n+1\) binomial coefficients sum to \(2^n\), one is at least \(2^n/(n+1)\).
\end{proof}

The first lemma restricts a candidate ladder ratio through its exponent.
The second shows that a short interval of such restrictions cannot all divide a small residual exponent.
Together they let the construction exclude every incorrect ratio while writing only logarithmically many exclusions at a stage.

\subsection{Power-of-Two Scaling and Bounded Products}
\label{subsec:succinct-powers}

The next construction defines multiplication by \(2^{2^t}-1\) using polynomially many constraints.
Its variables form a chain whose successive ratios are forced to be the Fermat numbers \(2^{2^h}+1\).
Some atoms restrict a ratio to a Fermat number, and additional atoms exclude the incorrect choices.

As previewed in \cref{subsec:overview-divisibility}, writing the final multiplier would already take exponential space.
The ladder uses polynomial-bit coefficients to constrain the successive ratios and rules out the spurious ones without enumerating all earlier stages.

\begin{lemma}[Mersenne multiplier]\label{lem:mersenne-multiplier}
For every integer \(t\ge0\), there is a polynomial-size positive homogeneous divisibility system \(\Delta_t(x,r)\) whose \((x,r)\)-projection is \(\{(x,(2^{2^t}-1)x):x\in\Npos\}\).
Adding \(v=r+x\) gives projection \(\{(x,2^{2^t}x):x\in\Npos\}\).
\end{lemma}

\begin{proof}
Use variables \(0<y_0<y_1<\cdots<y_t\), with \(y_0=x\) and \(y_t=r\).
When \(t=0\), these conditions reduce to \(r=x\).
For \(i=t,t-1,\ldots,1\), put \(h=t-i\) and include
\begin{align}
 y_i+y_{i-1}&\mid y_t+3y_i,\label{eq:stage-merge}\\
 y_{i-1}&\mid y_i,\label{eq:stage-a}\\
 y_i-y_{i-1}&\mid y_t+y_i,\label{eq:stage-b}\\
 y_i-2y_{i-1}&\mid y_t.\label{eq:stage-c}
\end{align}
For \(h\ge1\), let \(m_h\) be the least nonnegative integer satisfying \(2^{m_h+1}/(m_h+2)>h\), put \(D_h=\min(m_h,h-1)\), and, for \(0\le q\le D_h\), include
\begin{equation}\label{eq:exclusion-atom}
 y_i-(2^{2^q}+1)y_{i-1}
 \ \mid\
 y_t+(1-2^{2^q})y_i.
\end{equation}
Atom~\eqref{eq:stage-merge} is not needed for the projection asserted in this lemma.
It permits a finite-quotient replacement that fixes \(y_{i-1}\) relative to \(y_i\) once \(y_t/y_i\) is fixed.
Indeed, if \(y_t=cy_i\) with \(c>0\), then its quotient satisfies
\[
 q=\frac{(c+3)y_i}{y_i+y_{i-1}},
 \qquad 0<q<c+3,
\]
and the quotient equation gives \(y_{i-1}=(c+3-q)y_i/q\).
Thus there are finitely many integer quotients, and each consistent choice fixes the required ratio.

\proofparagraph{Soundness}
We prove, for \(i=t,t-1,\ldots,1\), that \(y_t=(B-1)y_i\), where \(B=2^{2^{t-i}}\), implies \(y_i=(B+1)y_{i-1}\).
The hypothesis is an identity when \(i=t\).
By \eqref{eq:stage-a}, write \(y_i=Qy_{i-1}\) with \(Q\ge2\).
Equation~\eqref{eq:stage-b} gives \(Q-1\mid BQ\), and \(\gcd(Q-1,Q)=1\), so \(Q-1\mid B\).
Thus \(Q=2^s+1\) for some \(s\ge0\).
The case \(s=0\) contradicts \eqref{eq:stage-c}.
Hence \(s>0\).
From \eqref{eq:stage-c} and \(\gcd(2^s-1,2^s+1)=1\), we obtain \(2^s-1\mid B-1=2^{2^{t-i}}-1\).
By \cref{lem:mersenne-div}, \(s\mid2^{t-i}\), so \(s=2^a\) for some \(0\le a\le h=t-i\), and \(Q=2^{2^a}+1\).
If \(h=0\), then \(a=h\).

Assume \(h>0\) and \(a<h\).
If \(a\le D_h\), the exclusion atom with \(q=a\) has zero left-hand side and positive right-hand side, a contradiction.
Thus \(a>D_h\).
Fix \(0\le q\le D_h\).
The coefficient \(2^{2^a}-2^{2^q}\) is coprime to \(2^{2^a}+1\).
Indeed, if an odd prime \(p\) divided both, then \(2^{2^a}\equiv-1\pmod p\), so the multiplicative order of \(2\) modulo \(p\) would be \(2^{a+1}\).
The congruence \(2^{2^a}\equiv2^{2^q}\pmod p\) would instead make that order divide \(2^a-2^q=2^q(2^{a-q}-1)\), which is impossible.
The exclusion atom therefore implies \(2^{2^a}-2^{2^q}\mid2^{2^h}-2^{2^q}\).
Cancelling the common factor \(2^{2^q}\) gives \(2^{2^a-2^q}-1\mid 2^{2^h-2^q}-1\), and \cref{lem:mersenne-div} yields \(2^a-2^q\mid 2^h-2^q\).
After cancelling \(2^q\), another application of \cref{lem:mersenne-div} gives \(a-q\mid h-q\), and hence \(a-q\mid h-a\).
If \(D_h=h-1\), then \(a>D_h\) contradicts \(a<h\).
Otherwise \(D_h=m_h\), and every integer in \(a-m_h,\ldots,a\) divides the positive integer \(h-a<h\).
By \cref{lem:lcm-interval}, their least common multiple exceeds \(h\), a contradiction.
Thus \(a=h\), and \(y_i=(B+1)y_{i-1}\) and \(y_t=(B^2-1)y_{i-1}\).
After stage \(i=1\), this gives \(y_t=(2^{2^t}-1)y_0\).

\proofparagraph{Completeness}
Conversely, fix \(x=y_0>0\) and set \(y_i=(2^{2^{t-i}}+1)y_{i-1}\) for \(1\le i\le t\).
The identity \(\prod_{q=0}^{h-1}(2^{2^q}+1)=2^{2^h}-1\) shows that, at stage \(i\), with \(h=t-i\) and \(B=2^{2^h}\), one has \(y_t=(B-1)y_i\) and \(y_i=(B+1)y_{i-1}\).
The quotient in \eqref{eq:stage-merge} is then \((y_t+3y_i)/(y_i+y_{i-1})=B+1\).
The quotients in \eqref{eq:stage-a}--\eqref{eq:stage-c} are also \(B+1\).
In an exclusion atom, both sides contain the factor \(2^{2^h}-2^{2^q}\), and the quotient is \(2^{2^h}+1\).
Thus every positive \(x\) extends to a solution.

\proofparagraph{Size}
The formula has polynomial size.
For \(h\ge1\), the integer \(\lceil2\log_2(h+2)\rceil+2\) satisfies the defining inequality for \(m_h\), so \(m_h=O(\log h)\).
There are \(O(t\log t)\) exclusion atoms, and the largest coefficient \(2^{2^q}\) has polynomial binary length.
\end{proof}

\begin{remark}\label{rem:mersenne-exclusions}
At stage \(h\), explicitly excluding every incorrect Fermat ratio would require coefficients \(2^{2^q}\) for values of \(q\) close to \(h\).
Those coefficients have exponential binary length in \(h\), even though there are only \(h\) possible ratios.
The construction writes exclusions only for \(q=O(\log h)\), where the coefficients have polynomial binary length.
The interval-lcm bound in \cref{lem:lcm-interval} then rules out all remaining smaller ratios simultaneously.
\end{remark}

For an arbitrary binary exponent, compose one ladder for each set bit and multiply the corresponding factors along a chain.
This gives a polynomial-size scaling system without writing the resulting power of two.

\begin{lemma}[Succinct positive power scaling]\label{lem:positive-scaling}
Given \(e\ge0\) in binary, there is a polynomial-size positive homogeneous divisibility system whose \((u,v)\)-projection is \(\{(u,2^eu):u\in\Npos\}\).
\end{lemma}

\begin{proof}
For \(e=0\), use \(v=u\).
Assume \(e>0\) and let \(J=\{j:e_j=1\}\) for the binary expansion \(e=\sum_{j\in J}2^j\).
Put \(s=|J|\), enumerate \(J=\{j_0,\ldots,j_{s-1}\}\), and introduce a chain \(u=z_0,z_1,\ldots,z_s=v\).
At step \(a\in\{0,\ldots,s-1\}\), take a fresh copy of \cref{lem:mersenne-multiplier} with endpoint \(r_a\), forcing \(r_a=(2^{2^{j_a}}-1)z_a\), and add \(z_{a+1}=r_a+z_a\).
The chain forces \(v=2^eu\).
Since \(s=|J|\le\len(e)\), there are at most \(\len(e)\) steps, each of size polynomial in \(\len(e)\).
\end{proof}

\begin{lemma}[Signed power-of-two scaling]\label{lem:signed-scaling}
Given \(e\ge0\) in binary, there is a polynomial-size EPAD formula \(\Scale_e(A,B)\) over free integer variables \(A,B\) such that \(\Scale_e(A,B)\) holds exactly when \(B=2^eA\).
\end{lemma}

\begin{proof}
Introduce positive variables \(A^+,A^-,B^+,B^-\), impose \(A=A^+-A^-\) and \(B=B^+-B^-\), and use two copies of \cref{lem:positive-scaling} to force \(B^+=2^eA^+\) and \(B^-=2^eA^-\).
Every integer is a difference of two positive integers, so the construction is complete as well as sound.
\end{proof}

The positive version keeps every variable positive and every constraint homogeneous.
The signed version permits arbitrary integer inputs, including zero, as required for auxiliary calculations in the circuit translation.

A bounded product can be defined by two divisibility atoms that share a quotient.
The bounds ensure that the second atom fixes this quotient to be the other operand, rather than merely constraining it modulo the divisor.

\begin{lemma}[Quotient-sharing product]\label{lem:product-gadget}
Let \(U,V,W\) be integer affine forms satisfying \(|U|<M_U\) and \(|W|<M_W\).  Let \(h\ge0\) be an integer and put \(\Lambda=2^h\), with \(\Lambda>M_U+M_W+1\).
Multiples by \(\Lambda\) are represented by fresh variables constrained with \cref{lem:signed-scaling}.
Under these bounds, the atoms
\begin{equation}\label{eq:quotient-sharing}
 V\mid W,
 \qquad
 1+\Lambda V\mid U+\Lambda W
\end{equation}
are equivalent to \(W=UV\).
\end{lemma}

\begin{proof}
If \(W=UV\), both atoms hold.
Conversely, if \(V=0\), then \(V\mid W\) forces \(W=0\).
If \(V\ne0\), write \(W=qV\).
The bound on \(W\) gives \(|q|<M_W\).
Modulo \(1+\Lambda V\), the second atom gives \(1+\Lambda V\mid U-q\).
Since \(|U-q|<M_U+M_W<\Lambda-1\le|1+\Lambda V|\), we obtain \(U=q\) and \(W=UV\).
\end{proof}

The gadget expresses the only nonlinear gate relation by equating two divisibility quotients.
Its numerical bounds are essential: without them the second atom would determine \(U-q\) only modulo \(1+\Lambda V\).
For circuit translation, the declared word-width supplies those bounds uniformly.

\section{Arithmetic Logic Circuits: \texorpdfstring{\(\FPSPACE\)}{FPSPACE} and \texorpdfstring{\(\FEXP\)}{FEXP}}
\label{sec:poly-alc}

To prove \(\polyALC\subseteq\FPSPACE\), we evaluate a requested output bit using polynomial space.
For the converse inclusion, we convert a polynomial-space indexed-bit predicate to a QBF and represent its truth table in a large integer.
Arithmetic operations on this representation evaluate all index assignments simultaneously.

\subsection{From Polynomial-Size ALCs to \texorpdfstring{\(\FPSPACE\)}{FPSPACE}}
\label{subsec:alc-to-fpspace}

A requested output bit is recovered by a polynomial-space recursion that never materializes an entire exponentially wide word.

\begin{proposition}\label{prop:alc-to-fpspace}
Every polynomial-time uniform polynomial-size ALC family computes a function in \(\FPSPACE\).
\end{proposition}

\begin{proof}
Fix an input of length \(n\), write \(W=W(n)\), and fix a requested output-bit position \(i<W\).
Whenever the recursion reaches a gate index, obtain its local description from the uniformity algorithm of \cref{def:alc}, and evaluate the requested gate bit recursively.
Input and constant bits are read directly.
A left shift by offset \(s\) replaces index \(i\) by \(i-s\), returning zero if \(i<s\).
A right shift by offset \(s\) replaces it by \(i+s\), returning zero if \(i+s\ge W\).

For addition modulo \(2^W\), scan columns zero through \(i\) and maintain the one-bit carry.
Higher columns cannot affect bit \(i\).
For a multiplication gate with predecessor words \(u,v\), let \(u_j,v_j\) denote their recursively obtained bits.  At column \(t\), compute \(w_t=\sum_{j=0}^{t}u_jv_{t-j}\) and maintain \(c_0=0\) and \(c_{t+1}=\lfloor(c_t+w_t)/2\rfloor\).
Since \(w_t\le t+1\), induction gives \(c_t\le t\), so the carry has \(O(\log(t+1))\) bits.
The requested product bit is \((c_i+w_i)\bmod 2\).

Gate indices, bit indices, and carries have polynomial length, and the recursion stack has depth at most the number of gates.
Each output bit is therefore computable in polynomial space.
The machine writes the \(W\) output bits from most significant to least significant, using the binary representation of \(W\) as its counter bound.
\end{proof}

The recursion may recompute a predecessor bit many times and therefore take exponential time, but it uses polynomial space because it never stores an intermediate \(W\)-bit word.
Thus exponential word-width does not itself require exponential workspace.

\subsection{From \texorpdfstring{\(\FPSPACE\)}{FPSPACE} to
Polynomial-Size ALCs by Packed Truth Tables}
\label{subsec:fpspace-to-alc}

A packed truth table places the values of a Boolean function at widely separated bit positions.
The proof needs to retain, from two packed tables, the products of entries belonging to the same assignment.
We use this operation for pointwise conjunction, restriction of a variable, and the final placement of the output bits.

The input and output tables of an operation may use different scales.
This freedom is essential because each multiplication creates cross terms around the entries that must survive.
We choose the input scale after the desired output positions are known.
An explicit separation bound keeps cross terms and carries outside the retained block.

For \(\vec a=(a_0,\ldots,a_{k-1})\in\bits^k\), let \(\enc(\vec a)=\sum_{i=0}^{k-1}a_i5^i\), the integer whose base-five digits are \(a_0,\ldots,a_{k-1}\), from least to most significant.
For a Boolean function \(F\colon\bits^k\to\bits\) and a positive integer \(R\), define its packed truth table at scale \(R\) by
\[
 \langle F\rangle_R
 =\sum_{\vec a\in\bits^k}F(\vec a)2^{R\enc(\vec a)}.
\]
Thus the bit at position \(R\enc(\vec a)\) is \(F(\vec a)\).

Let \(F,G\colon\bits^k\to\bits\), let \(R>0\) be an integer, and let \(\vec\delta=(\delta_0,\ldots,\delta_{k-1})\in\N^k\).
Define
\begin{equation}\label{eq:matching-product}
 \mathcal M_{R,\vec\delta}(F,G)
 =\langle F\rangle_R\langle G\rangle_R
  \prod_{i=0}^{k-1}\left(2^{2R5^i}+2^{\delta_i}\right).
\end{equation}
This product contains terms arising from every pair of assignments to \(F\) and \(G\).
The next lemma identifies the consecutive bits that retain exactly the pairs in which the two assignments are equal.
For an integer \(x\), write \(\operatorname{block}_{h,d}(x)= \lfloor x/2^h\rfloor\bmod 2^d\) for the \(d\) consecutive bits of \(x\) beginning at position \(h\).
\begin{lemma}[Diagonal extraction]\label{lem:diagonal-extraction}
Let \(F,G\colon\bits^k\to\bits\), let \(R>0\) be an integer, let \(\vec\delta\in\N^k\), and put \(h=R(5^k-1)/2\).
If \(R>\sum_{i=0}^{k-1}\delta_i+3k+1\), then
\[
 \operatorname{block}_{h,R}
   \bigl(\mathcal M_{R,\vec\delta}(F,G)\bigr)
 =\sum_{\vec a\in\bits^k}F(\vec a)G(\vec a)
  2^{\sum_{i=0}^{k-1}a_i\delta_i}.
\]
\end{lemma}

\begin{proof}
Expand \(\mathcal M_{R,\vec\delta}(F,G)\).
A term is determined by \(\vec a,\vec b,\vec z\in\bits^k\), where \(\vec a\) and \(\vec b\) select entries from the two packed tables and \(z_i\) records the choice of \(2^{2R5^i}\) from the \(i\)-th factor in \eqref{eq:matching-product}.
Its exponent is \(Rq+r\), where
\[
 q=\sum_{i=0}^{k-1}(a_i+b_i+2z_i)5^i,
 \qquad
 r=\sum_{i=0}^{k-1}(1-z_i)\delta_i.
\]
Every base-five digit of \(q\) lies in \(\{0,1,2,3,4\}\), so these digits produce no carries.
Since \((5^k-1)/2=\sum_{i=0}^{k-1}2\cdot5^i\), we have \(Rq=h\) exactly when \(a_i=b_i\) and \(z_i=1-a_i\) for every \(i\).
For those terms, \(r=\sum_i a_i\delta_i\), and their sum is the right-hand side of the claimed identity.

It remains to check that no other term contributes to the block beginning at \(h\).
If \(Rq<h\), then the exponent is at most \(h-R+\sum_i\delta_i\).
There are at most \(2^{3k}\) terms, and the assumed inequality prevents their sum from carrying into position \(h\).
If \(Rq>h\), then the exponent is at least \(h+R\), beyond the extracted block.
Finally, the retained sum is at most \(2^{k+\sum_i\delta_i}<2^R\), so it fits within that block.
\end{proof}

The lemma is an arithmetic analogue of projecting a matrix onto its diagonal.
The base-five address identifies the diagonal terms, while the inequality on \(R\) protects the corresponding binary block from neighboring terms and carries.
The construction uses only multiplication and shifts available in an ALC.

\begin{corollary}[Pointwise product]\label{lem:pointwise}
Let \(F,G\colon\bits^k\to\bits\), let \(R,R'>0\) be integers, define \(\vec\delta\in\N^k\) by \(\delta_i=R'5^i\), and put \(h=R(5^k-1)/2\).
If \(R>R'(5^k-1)/4+3k+1\), then
\[
 \operatorname{block}_{h,R}
   \bigl(\mathcal M_{R,\vec\delta}(F,G)\bigr)
 =\langle F\wedge G\rangle_{R'}.
\]
\end{corollary}

\begin{proof}
In \Cref{lem:diagonal-extraction}, the exponent attached to \(\vec a\) is \(R'\enc(\vec a)\), and \(F(\vec a)G(\vec a)=(F\wedge G)(\vec a)\).
\end{proof}

\begin{corollary}[Restriction]\label{lem:restriction}
Let \(F\colon\bits^{k+1}\to\bits\), let \(R,R'>0\) be integers, write its last argument as \(y\), and fix \(b\in\bits\).
Put \(F_b(\vec a)=F(\vec a,b)\), and let \(\chi_b(\vec a,y)\) be one when \(y=b\) and zero otherwise.
Define \(\vec\delta\in\N^{k+1}\) by \(\delta_i=R'5^i\) for \(i<k\) and \(\delta_k=0\), and put \(h=R(5^{k+1}-1)/2\).
If \(R>R'(5^k-1)/4+3(k+1)+1\), then
\[
 \operatorname{block}_{h,R}
   \bigl(\mathcal M_{R,\vec\delta}(F,\chi_b)\bigr)
 =\langle F_b\rangle_{R'}.
\]
\end{corollary}

\begin{proof}
The packed table used in \(\mathcal M_{R,\vec\delta}(F,\chi_b)\) is
\[
 \langle\chi_b\rangle_R
 =\prod_{i=0}^{k-1}(1+2^{R5^i})
 \begin{cases}
  1,&b=0,\\
  2^{R5^k},&b=1.
 \end{cases}
\]
Apply \Cref{lem:diagonal-extraction}.
The factor \(\chi_b(\vec a,y)\) removes the terms with \(y\ne b\), and \(\delta_k=0\) removes the fixed coordinate from the output positions.
\end{proof}

Pointwise product supplies conjunction of packed predicates.
Restriction supplies substitution of a quantified variable by each Boolean value.
OR is obtained arithmetically from the two restricted Boolean tables.
These are the operations needed to evaluate a quantified formula from its innermost subformula outward.

For the converse inclusion, an indexed-bit predicate is first represented by a polynomial-size QBF\@.
The circuit builds one packed table for each subformula, eliminates quantified variables by restriction, and finally extracts the block containing the output bits.
The remaining obligations are to choose all scales uniformly and to keep the gate list polynomial.

\begin{theorem}\label{thm:fpspace-to-alc}
Every function in \(\FPSPACE\) is computed by a polynomial-time uniform family of polynomial-size ALCs.
\end{theorem}

\begin{proof}
We first describe the arithmetic calculation, then choose a common word-width \(W\) large enough for its intermediate values.
Ordinary sums and products are realized without overflow.
Bit extraction and complementation of Boolean inputs use fixed-width shifts, with truncation made explicit.
For any nonnegative exponent \(E\) needed by the construction, the circuit forms \(2^E\) as \(1\ll_W E\).
For nonnegative integers \(h,d\) with \(h+d\le W\), it extracts the \(d\)-bit block of an intermediate value \(x\) beginning at position \(h\) by
\begin{equation}\label{eq:block-extraction}
 \bigl((x\gg h)\ll_W(W-d)\bigr)\gg(W-d).
\end{equation}

Let \(f=(f_n)_n\in\FPSPACE\), with \(f_n\colon\bits^n\to\N\).
By \Cref{lem:indexed-bits}, fix a polynomial \(p\) and the corresponding \(\PSPACE\) indexed-bit predicate.
Using the \(p(n)\)-bit binary representation of the index, with the least significant bit first, write this predicate as \(\mathcal B(\vec x,\vec a)\), where \(\vec a\in\bits^{p(n)}\).

The polynomial-space-to-QBF reduction \cite[Chapter~19]{Papadimitriou1994} gives, uniformly in \(n\), a polynomial-size prenex quantified Boolean formula with free variables \(\vec x,\vec a\) and \(q=\poly(n)\) quantified variables.  Writing \(\vec y=(y_0,\ldots,y_{q-1})\), it has the form
\[
 \Psi_n(\vec x,\vec a)
 =Q_0y_0\cdots Q_{q-1}y_{q-1}\,\varphi_n(\vec x,\vec a,\vec y).
\]
It is true exactly when \(\mathcal B(\vec x,\vec a)=1\).
We pack the variables in the order \(\vec a,y_0,\ldots,y_{q-1}\).
The variables \(\vec x\) remain ALC inputs.
We eliminate \(y_{q-1},\ldots,y_0\) in this order, so the variable removed at each step is the last variable indexing the current table.

For a subformula \(\theta\) and \(b\in\bits\), let \(F_\theta^b\) be the Boolean function of the variables indexing the current table, with \(\vec x\) treated as parameters, that equals \(1\) exactly when \(\theta\) has truth value \(b\).
The construction produces \(\langle F_\theta^0\rangle_R\) and \(\langle F_\theta^1\rangle_R\) at any requested output scale \(R\).
At a connective or quantifier, it chooses one common larger scale for the operand tables, large enough to satisfy the relevant inequality in \Cref{lem:pointwise,lem:restriction}.
Taking the right-hand side plus one suffices.
Hence the operand tables used at one step all have the same scale.

\proofparagraph{Literals}
Suppose the current table is indexed by \(v_0,\ldots,v_{k-1}\).
Its all-ones table at scale \(R\) is
\[
 A_{k,R}=\sum_{\vec\xi\in\bits^k}2^{R\enc(\vec\xi)}
 =\prod_{i=0}^{k-1}(1+2^{R5^i}).
\]
For the literal \(v_j\),
\[
 \langle F_{v_j}^1\rangle_R
 =2^{R5^j}\prod_{i\ne j}(1+2^{R5^i}),
 \qquad
 \langle F_{v_j}^0\rangle_R
 =\prod_{i\ne j}(1+2^{R5^i}).
\]
For an ALC input bit \(x_t\), put \(\bar x_t=((x_t+_W 1)\ll_W(W-1))\gg(W-1)\).
Since \(x_t\in\bits\), this value is \(1-x_t\).
Thus
\[
 \langle F_{x_t}^1\rangle_R=x_tA_{k,R},
 \qquad
 \langle F_{x_t}^0\rangle_R=\bar x_tA_{k,R}.
\]
\proofparagraph{Boolean connectives}
The characteristic functions satisfy
\[
 \begin{aligned}
 F_{\alpha\wedge\beta}^1&=F_\alpha^1F_\beta^1,
 &
 F_{\alpha\wedge\beta}^0
 &=F_\alpha^0F_\beta^0+F_\alpha^0F_\beta^1+F_\alpha^1F_\beta^0,\\
 F_{\alpha\vee\beta}^0&=F_\alpha^0F_\beta^0,
 &
 F_{\alpha\vee\beta}^1
 &=F_\alpha^0F_\beta^1+F_\alpha^1F_\beta^0+F_\alpha^1F_\beta^1.
 \end{aligned}
\]
Products here are pointwise products of Boolean functions.
Each is implemented by \Cref{lem:pointwise} at the requested output scale.
For each assignment, at most one summand in either three-term sum is \(1\).
Adding the corresponding packed tables therefore produces no carry and gives the packed table of the sum.
Negation satisfies \(F_{\neg\alpha}^b=F_\alpha^{1-b}\).

\proofparagraph{Quantifiers}
Suppose \(y\) is the last variable indexing the table for \(\theta\).
\Cref{lem:restriction} gives the packed tables of \(\vec\xi\mapsto F_\theta^c(\vec\xi,b)\) for \(b,c\in\bits\).
The required characteristic functions are
\[
 \begin{aligned}
 F_{\exists y\,\theta}^0(\vec\xi)
 &=F_\theta^0(\vec\xi,0)F_\theta^0(\vec\xi,1),\\
 F_{\exists y\,\theta}^1(\vec\xi)
 &=F_\theta^0(\vec\xi,0)F_\theta^1(\vec\xi,1)
  +F_\theta^1(\vec\xi,0)F_\theta^0(\vec\xi,1)
  +F_\theta^1(\vec\xi,0)F_\theta^1(\vec\xi,1),\\
 F_{\forall y\,\theta}^1(\vec\xi)
 &=F_\theta^1(\vec\xi,0)F_\theta^1(\vec\xi,1),\\
 F_{\forall y\,\theta}^0(\vec\xi)
 &=F_\theta^0(\vec\xi,0)F_\theta^0(\vec\xi,1)
  +F_\theta^0(\vec\xi,0)F_\theta^1(\vec\xi,1)
  +F_\theta^1(\vec\xi,0)F_\theta^0(\vec\xi,1).
 \end{aligned}
\]
As in the connective case, \Cref{lem:pointwise} implements each product.
For each assignment, at most one summand in a three-term sum is \(1\), so adding the packed tables produces no carry.

If a table indexed by \(k\) variables is requested at scale \(R'\), either inequality in \Cref{lem:pointwise,lem:restriction} can be met with a scale \(R\) satisfying \(\log R\le\log R'+O(k)\).
The QBF has polynomial size, so the recursion has polynomial depth and \(k=\poly(n)\) at every step.
Consequently every scale used by the construction has a polynomial-length binary representation and is computable in polynomial time.

\proofparagraph{Output}
Choose the requested scale of the final table so that \(R>2^{p(n)}+3p(n)\).
After eliminating all quantified variables, the construction gives
\[
 \langle F_{\Psi_n}^1\rangle_R
 =\sum_{\vec a\in\bits^{p(n)}}
 \mathcal B(\vec x,\vec a)2^{R\enc(\vec a)}.
\]
Define \(\vec\delta\in\N^{p(n)}\) by \(\delta_j=2^j\), and put \(h=R(5^{p(n)}-1)/2\).
By \Cref{lem:diagonal-extraction},
\[
 \operatorname{block}_{h,R}
 \bigl(\mathcal M_{R,\vec\delta}(F_{\Psi_n}^1,F_{\Psi_n}^1)\bigr)
 =\sum_{\vec a\in\bits^{p(n)}}
 \mathcal B(\vec x,\vec a)2^{\sum_{j=0}^{p(n)-1}a_j2^j}
 =f_n(\vec x).
\]
The output has at most \(2^{p(n)}\) bits by \Cref{lem:indexed-bits}.
The circuit therefore extracts the \(2^{p(n)}\)-bit block beginning at \(h\) using \eqref{eq:block-extraction}.

Each subformula occurrence is constructed once at its requested scale, and forming \(\mathcal M_{R,\vec\delta}(F,G)\) uses \(O(k)\) gates.
Since the QBF has polynomial size, the resulting ALC has polynomial size.

It remains to choose \(W\).
All exponents, scales, and block endpoints used to describe the calculation have polynomial-length binary representations.
Along the polynomial-size acyclic calculation, addition increases bit length by at most one, multiplication adds the bit lengths of its inputs, and a left shift by \(E\) increases bit length by at most \(E\).
Therefore a value \(W\le2^{\poly(n)}\), with a polynomial-length binary representation, can be chosen larger than the bit length of every intermediate integer and larger than every block endpoint.
With this word-width, all preceding calculations are valid ALC operations.
The QBF reduction, the recursive scale choices, and the gate construction are polynomial-time uniform.
\end{proof}

The backward scale choices attach a polynomial-length binary parameter to each subformula occurrence without creating a gate for every truth-table entry.
Consequently the circuit has polynomially many gates even though its arithmetic values encode exponentially many entries.

\AlcFpspace*
\begin{proof}
Combine \cref{prop:alc-to-fpspace,thm:fpspace-to-alc}.
\end{proof}

The following standard fixed-width arithmetic and logic operations do not change \(\polyALC\).
They are definable from the four primitives.

\begin{proposition}[Completeness for standard fixed-width operations]
\label{prop:alc-standard-operations}
Let \(\polyALCstd\) be defined like \(\polyALC\), but with the following additional gate types: fixed-width subtraction, two's-complement negation, and unsigned quotient and remainder.
It also has bitwise negation, conjunction, disjunction, and exclusive disjunction, unsigned equality and order comparisons, and logical shifts and left or right rotations whose offset may be supplied by another operand.
Arithmetic results are reduced modulo \(2^W\), comparisons return \(0\) or \(1\), shifts by at least \(W\) return zero, and rotations use the offset modulo \(W\).
Under any fixed convention for division by zero, \(\polyALCstd=\polyALC\).
\end{proposition}

\begin{proof}
Every circuit over the original basis is also a circuit over the extended basis.
For the converse, consider a polynomial-time uniform polynomial-size family over the extended basis and, at input length \(n\), write \(W=W(n)\).
Standard locally uniform Boolean circuits of size \(W^{O(1)}\) and depth \((\log W)^{O(1)}\) implement every additional gate.
For division and remainder, use the uniform threshold circuits of Hesse, Allender, and Barrington~\cite{HesseAllenderBarrington2002,HesseAllenderBarrington2014}.
Recursively evaluating one requested output bit stores only polynomially many gate indices and counters, each of length \(O(\log W)\).
The enclosing ALC has polynomially many gates and \(\log W=\poly(n)\), so its indexed-output-bit predicate is in \(\PSPACE\).
The output has at most \(W\) bits.
Hence \Cref{lem:indexed-bits,thm:alc-fpspace} give an equivalent polynomial-time uniform polynomial-size ALC family over the original basis.
\end{proof}

\subsection{Exponential-Size ALCs and \texorpdfstring{\(\FEXP\)}{FEXP}}
\label{sec:exp-alc}

The characterization above uses the polynomial-size class \(\polyALC\).
The class \(\expALC\) from \cref{def:alc} changes only the size bound: the ALC operations, exponential word-width, and polynomial-time local uniformity condition are identical, but \(C_n\) may now have \(2^{\poly(n)}\) gates.
Thus no second notion of uniformity or circuit representation is introduced.

Like \(\FPSPACE\), the class \(\FEXP\) permits outputs of exponential length.
On inputs of length \(n\), an output may contain \(2^{\poly(n)}\) bits and thus represent an integer as large as \(2^{2^{\poly(n)}}-1\).
The additional computational power of \(\expALC\) comes from the exponentially larger number of gates, not from a wider word or a different uniformity condition.

\begin{proposition}\label{prop:expalc-fexp}
Polynomial-time uniform exponential-size ALC families compute exactly \(\FEXP\):
\[
 \expALC=\FEXP.
\]
\end{proposition}
\begin{proof}
\proofparagraph{From \(\expALC\) to \(\FEXP\)}
Let \((C_n)\) be an \(\expALC\) family, so \(S(n),W(n)\le2^{\poly(n)}\).
Compute \(S(n)\) and enumerate the gate indices \(0,\ldots,S(n)-1\) in topological order.
For each gate, query the uniformity algorithm for its local description and evaluate it from the stored predecessor values.
An addition, a shift, or a schoolbook multiplication on \(W(n)\)-bit words takes time polynomial in \(W(n)\).
Hence all \(S(n)\) gates are evaluated in \(2^{\poly(n)}\) time, proving \(\expALC\subseteq\FEXP\).

\proofparagraph{From \(\FEXP\) to \(\expALC\)}
Conversely, let a deterministic transducer run for \(T(n)\le2^{\poly(n)}\) steps.
The standard tableau construction gives a bounded-fan-in Boolean circuit of size \(T(n)^{O(1)}\) for its computation: a gate is indexed by a time, tape position, track, and a constant-size local subgate.
Its type and predecessors are computable from that index in polynomial time.
Include \(T(n)\) output cells, padded by leading zeros, and choose an ALC word-width \(W=W(n)\ge T(n)\).

A Boolean AND gate is the product of its one-bit inputs.
Boolean negation is computed by
\[
 \neg a=\bigl((a+_W 1)\ll_W(W-1)\bigr)\gg(W-1),
 \qquad a\in\bits.
\]
Thus every bounded-fan-in Boolean gate has a constant-size ALC simulation.
Pack the output bits \(b_0,\ldots,b_{W-1}\) as \(\sum_i(b_i\ll_W i)\) using a balanced addition tree.
The summands have disjoint supports, so the sum is the required word.
The resulting ALC has exponential size, and its local gate description remains computable in polynomial time.
Therefore \(\FEXP\subseteq\expALC\).
\end{proof}

Both inclusions use the same polynomial-time local-uniformity interface as \(\polyALC\).
At polynomial size that interface can be enumerated to print the whole circuit in polynomial time; at exponential size the same enumeration takes exponential time.

\begin{remark}[Whole-circuit uniformity]
The same proof also shows that \(\FEXP\) is exactly the class computed by exponential-size ALC families whose whole circuits can be generated in exponential time.
The stronger polynomial-time local uniformity of \(\expALC\) lets a polynomial-size Boolean circuit return any addressed clause of a Boolean formula expressing correct circuit evaluation.
This gives a succinct description of the evaluation constraints, whereas writing one constraint for each gate would take exponential space.
Section~\ref{sec:expalc-fepad-section} uses this description to obtain polynomial-size functional EPAD definitions.
\end{remark}

\section{Functional EPAD Compilation of ALCs}
\label{sec:poly-alc-fepad}

We translate each gate of a polynomial-size ALC into EPAD\@.
Addition uses a congruence, and right shift uses a quotient--remainder equation.
The other constructions needed for the translation define multiplication by a power of two whose exponent is given in binary and define the bounded product of two variables.
The resulting translation proves \(\polyALC\subseteq\FEPAD\).

The translation introduces one bounded integer variable for each circuit gate.
Once the predecessor values are fixed, the constraints for a gate determine its value uniquely.

The proof follows the circuit order.
At each gate it uses the scaling and product interfaces just established, then enforces the unique representative in the fixed-width output range.

\begin{theorem}\label{thm:alc-to-fepad}
For every ALC of word-width \(W\), one can construct an EPAD formula in time polynomial in the circuit description and in \(\len(W)\).
The formula defines the graph of the function computed by the circuit and is functional in the distinguished output.
Consequently, \(\polyALC\subseteq\FEPAD\).
\end{theorem}

\begin{proof}
Let \(C\) be an ALC of word-width \(W\)\@, and write its Boolean inputs as \(x_0,\ldots,x_{n-1}\).
Generate variables \(P=2^W\) and \(P_2=2^{2W}\) from the constant one by \cref{lem:signed-scaling}.
For every gate \(g\), introduce a value variable \(Y_g\) and impose \(0\le Y_g<P\).
Boolean input gates are identified with the corresponding \(x_i\) and constrained by \(0\le x_i\le1\).
Constant gates are fixed to zero or one.

\proofparagraph{Addition}
For an addition gate \(g=u+_W v\), impose \(P\mid Y_u+Y_v-Y_g\).
The bound \(0\le Y_g<P\) forces \(Y_g\) to be the unique remainder of \(Y_u+Y_v\) modulo \(P\).

\proofparagraph{Multiplication}
For a multiplication gate \(g=u\times_W v\), quantify \(T\) with \(0\le T<P_2\).
Apply \cref{lem:product-gadget} to force \(T=Y_uY_v\), using operand bound \(2^W\), product bound \(2^{2W}\), and \(\Lambda=2^{2W+2}\).
Then impose \(P\mid T-Y_g\).
The bound \(0\le Y_g<P\) makes \(Y_g\) the unique remainder of \(T\) modulo \(P\).

\proofparagraph{Left shift}
For a left shift \(g=u\ll_W s\), if \(s\ge W\), impose \(Y_g=0\).
Otherwise, quantify \(T\), use \(\Scale_s(Y_u,T)\), and impose \(P\mid T-Y_g\).
Since \(s<W\) and \(Y_u<P\), one has \(T<P_2\).

\proofparagraph{Right shift}
For a right shift \(g=u\gg s\), if \(s\ge W\), impose \(Y_g=0\).
Otherwise, generate \(P_s=2^s\), quantify \(R,Q\), and impose \(\Scale_s(Y_g,Q)\), \(Y_u=R+Q\), and \(0\le R<P_s\).
This is the Euclidean decomposition of \(Y_u\) by \(2^s\), and therefore forces \(Y_g=\lfloor Y_u/2^s\rfloor\).

Identify the formula output \(y\) with the value of the distinguished circuit gate.
Induction over the circuit shows that every gate value is uniquely determined.
Each gate contributes polynomially many constraints, and every exponent in the circuit description has polynomial binary length.
For a polynomial-time uniform polynomial-size family, the formula generator enumerates the polynomially many gate indices and obtains each local gate description from the uniformity algorithm of \cref{def:alc}; hence the resulting EPAD family is polynomial-time uniform.
\end{proof}

Functionality is stronger than equisatisfiability and is needed for composition.
An EPAD formula can therefore use the circuit output as an exponentially long constant without admitting a second value for the same inputs.
The gate-by-gate uniqueness argument preserves that property through the translation.

\begin{corollary}\label{cor:fpspace-fepad}
One has \(\FPSPACE\subseteq\FEPAD\).
\end{corollary}
\begin{proof}
Combine \cref{thm:alc-fpspace,thm:alc-to-fepad}.
\end{proof}

\begin{remark}\label{rem:expalc-fepad}
In \cref{sec:expalc-fepad-section}, we will prove that every function in \(\expALC\), the exponential-size class of \cref{def:alc}, has polynomial-size functional EPAD definitions of its graph.
Together with \(\expALC=\FEXP\), this will yield
\[
  \FEXP\subseteq\FEPAD.
\]
\end{remark}

\section{Encoding Succinct Exact-2-in-4 Satisfiability}
\label{sec:succinct-exact-two-normal-form}

This section proves \cref{thm:succinct-exact-two-normal-form}, which expresses the existence of a satisfying assignment using one signed integer witness and two divisibility conditions.

\begin{definition}[Succinct Exact-2-in-4 instance]
\label{def:succinct-exact-two}
A succinct Exact-2-in-4 instance is a triple
\[
 \mathcal I=(D,m,\ell),
\]
where \(m,\ell\ge1\) are written in binary and \(D\) is a Boolean \emph{clause-description circuit} for \(m\) ordered clauses over variables \(u_0,\ldots,u_{\ell-1}\).
Put
\[
 r_m=\max(1,\lceil\log_2m\rceil),
 \qquad
 r_\ell=\max(1,\lceil\log_2\ell\rceil).
\]
The circuit \(D\) has \(r_m\) input bits and \(4(r_\ell+1)\) output bits.
On the \(r_m\)-bit representation of a clause index \(0\le k<m\), it returns an \(r_\ell\)-bit integer \(j_h\) and a sign bit \(b_h\) for each \(h\in\{1,2,3,4\}\).
Put \(\varepsilon_h=+1\) when \(b_h=1\) and \(\varepsilon_h=-1\) when \(b_h=0\), and interpret \(i_h=j_h\bmod\ell\).
The pair \((i_h,\varepsilon_h)\) denotes \(u_{i_h}\) when \(\varepsilon_h=+1\) and \(\neg u_{i_h}\) otherwise.
These four literal occurrences form clause \(k\), with repetitions allowed and counted with multiplicity.
An assignment satisfies a clause exactly when two of its four literal occurrences are true.
Let \(F_{\mathcal I}\) denote the Exact-2-in-4 formula described by \(\mathcal I\), and let \(|\mathcal I|\) be the binary length of its description.

The problem \textsc{Succinct Exact-2-in-4-SAT} asks whether \(F_{\mathcal I}\) has a satisfying assignment; malformed encodings are rejected.
\end{definition}

\begin{lemma}[Local conversion from 3CNF to Exact-2-in-4]
\label{lem:three-cnf-to-exact-two}
Let \(G\) be a 3CNF formula with \(m\) clauses and \(\ell\) variables.
One can construct in polynomial time an equisatisfiable Exact-2-in-4 formula \(F\) with \(3m\) clauses and \(\ell+4m+1\) variables.
It contains the variables of \(G\) and a fresh reference variable \(t\).
An assignment to the variables of \(G\) satisfies \(G\) if and only if it extends to a satisfying assignment of \(F\) with \(t=1\).
Moreover, the conversion is local: from a polynomial-size clause-description circuit describing \(G\), one can construct in polynomial time a polynomial-size clause-description circuit describing \(F\), and the variables inherited from \(G\) retain their indices.
\end{lemma}
\begin{proof}
Write \(R(x,y,z,w)\) for the requirement that exactly two of the four displayed literal occurrences are true.
Introduce one reference variable \(t\), shared by all clauses.
For each clause \(x\vee y\vee z\) of \(G\), introduce four fresh variables \(a,b,c,d\) and the three clauses
\[
 R(\neg x,a,b,t),\qquad R(y,b,c,t),\qquad R(\neg z,c,d,t).
\]
Fix \(t=1\) and truth values \(x,y,z\in\bits\) for the three source literals.
The clauses are equivalent to
\[
 a+b=x,\qquad b+c=1-y,\qquad c+d=z.
\]
If \(y=1\), take \(b=c=0\), \(a=x\), and \(d=z\).
If \(y=0\) and \(x=1\), take \(b=1\), \(a=c=0\), and \(d=z\).
If \(x=y=0\), the first two equations force \(b=0\) and \(c=1\), so the third equation is solvable exactly when \(z=1\).
Thus the gadget with \(t=1\) has a satisfying extension exactly when \(x\vee y\vee z\) is true.
The auxiliary variables are private to each source clause, which proves the extension property for \(F\).

Complementing every variable preserves each Exact-2-in-4 clause.
Thus every satisfying assignment of \(F\) can be complemented, if necessary, to have \(t=1\), proving that \(F\) and \(G\) are equisatisfiable.

For the local statement, an output clause index \(k\) determines the original clause \(q=\lfloor k/3\rfloor\) and one of its three gadget clauses.
The source clause-description circuit is queried once; its returned indices are first reduced modulo the source variable count, as in \cref{def:succinct-exact-two}.
Give \(t\) index \(\ell\), and give the four auxiliary variables for source clause \(q\) indices \(\ell+1+4q,\ldots,\ell+4+4q\).
These indices are computed by binary arithmetic.
All indices have polynomial length in the source description, so the resulting clause-description circuit has polynomial size.
\end{proof}

\begin{proposition}\label{prop:succinct-exact-two-nexp}
\textsc{Succinct Exact-2-in-4-SAT} is \(\NEXP\)-complete under polynomial-time many-one reductions.
\end{proposition}
\begin{proof}
For membership, let \(\mathcal I=(D,m,\ell)\) be an instance and put \(N=|\mathcal I|\).
Since \(m,\ell\le2^N\), a nondeterministic exponential-time machine can guess an \(\ell\)-bit assignment, enumerate the \(m\) clauses, evaluate \(D\) on each clause index, and check the Exact-2-in-4 condition.
Thus the problem belongs to \(\NEXP\).

For hardness, start from the standard \(\NEXP\)-complete succinct 3CNF satisfiability problem~\cite[Section~20.1]{Papadimitriou1994}; see also~\cite{PapadimitriouYannakakis1986}.
In that representation, a polynomial-size circuit, given a clause index and one of its three literal positions, returns the corresponding variable index and sign of an exponentially large 3CNF; we use the equivalent convention in which the numbers of clauses and variables are included in binary in the description.
Three copies of this circuit in parallel return the three literals of an addressed source clause.
Apply \cref{lem:three-cnf-to-exact-two} to the described 3CNF\@.
Its local part constructs the Exact-2-in-4 clause-description circuit in polynomial time and preserves satisfiability.
\end{proof}

We now prove the normal form announced in \cref{subsec:overview-reduction}, restated here for convenience.

\OneWitnessNormalForm*

For the remainder of this section, fix a succinct Exact-2-in-4 instance \(\mathcal I=(D,m,\ell)\).

Write \(\vec\alpha=(\alpha_0,\ldots,\alpha_{\ell-1})\in\{-1,+1\}^{\ell}\), and let \(\vec u(\vec\alpha)\in\bits^\ell\) be the assignment
\[
 u_i(\vec\alpha)=\frac{1+\alpha_i}{2}
 \qquad(0\le i<\ell).
\]
The proof first packs all clause tests into one zero-target congruence, then replaces the \(\ell\) signs by one signed integer \(X\).
These constructions are combined in \cref{subsec:one-integer-normal-form} to prove the stated equivalence.

\subsection{Combining All Exact-2-in-4 Clauses into One Congruence}
\label{subsec:exact-two-congruence}

For \(0\le i<\ell\) and \(0\le k<m\), let \(d_{ik}\) be the number of positive occurrences of \(u_i\) in clause \(k\), minus the number of negative occurrences.
For \(0\le k<m\), define
\[
 s_k(\vec\alpha)
 =\sum_{i=0}^{\ell-1}d_{ik}\alpha_i.
\]
If \(n_k\) is the number of true literal occurrences in clause \(k\) under \(\vec u(\vec\alpha)\), then
\[
 s_k(\vec\alpha)=2n_k-4.
\]
Thus \(s_k(\vec\alpha)=0\) exactly when clause \(k\) has two true occurrences.

Define
\[
 \mathcal E(\vec\alpha)
 =\sum_{k=0}^{m-1}s_k(\vec\alpha)4^k,
\]
and put
\[
 M=2\cdot4^m=2^{2m+1}.
\]
The base-four spacing prevents cancellation between nonzero clause scores.

\begin{lemma}[Zero-target congruence encoding]
\label{lem:exact-two-congruence}
For every \(\vec\alpha\in\{-1,+1\}^{\ell}\),
\[
 \vec u(\vec\alpha)\models F_{\mathcal I}
 \quad\Longleftrightarrow\quad
 \mathcal E(\vec\alpha)=0
 \quad\Longleftrightarrow\quad
 \mathcal E(\vec\alpha)\equiv0\pmod M.
\]
\end{lemma}
\begin{proof}
If \(\vec u(\vec\alpha)\) satisfies \(F_{\mathcal I}\), then every score vanishes, so \(\mathcal E(\vec\alpha)=0\).
Conversely, suppose \(\mathcal E(\vec\alpha)=0\).
Each score belongs to \(\{-4,-2,0,2,4\}\), so
\[
 \frac{\mathcal E(\vec\alpha)}2
 =\sum_{k=0}^{m-1}\frac{s_k(\vec\alpha)}2\,4^k
\]
has coefficients in \(\{-2,-1,0,1,2\}\).
Reducing modulo \(4\) forces the coefficient at \(4^0\) to be zero; dividing by \(4\) and iterating shows that every score is zero.
The assignment \(\vec u(\vec\alpha)\) therefore satisfies \(F_{\mathcal I}\).

Finally,
\begin{equation}\label{eq:packed-score-bound}
 |\mathcal E(\vec\alpha)|
 \le4\sum_{k=0}^{m-1}4^k
 =\frac{4(4^m-1)}3
 <2\cdot4^m=M.
\end{equation}
Thus \(\mathcal E(\vec\alpha)\equiv0\pmod M\) if and only if \(\mathcal E(\vec\alpha)=0\).
\end{proof}

For \(0\le i<\ell\), define
\[
 c_i=\sum_{k=0}^{m-1}d_{ik}4^k.
\]
Then
\begin{equation}\label{eq:exact-two-linear-congruence}
 \mathcal E(\vec\alpha)=\sum_{j=0}^{\ell-1}c_j\alpha_j,
 \qquad
 F_{\mathcal I}\text{ satisfiable}
 \Longleftrightarrow
 \exists\vec\alpha\in\{-1,+1\}^{\ell}\colon
 \sum_{j=0}^{\ell-1}c_j\alpha_j\equiv0\pmod M.
\end{equation}

\subsection{Encoding the Sign Vector by One Integer}
\label{subsec:assignment-one-integer}

We adapt the square-difference sign encoding of Manders and Adleman~\cite[Lemma~1]{MandersAdleman1978}.

Put \(e=2(m+\ell)\).
Let \(p_0,\ldots,p_{\ell-1}\) be the first \(\ell\) odd primes, and define
\[
 K=\prod_{j=0}^{\ell-1}p_j^e,
 \qquad
 P_j=K/p_j^e.
\]
Thus \(P_j\) contains every prime-power factor of \(K\) except \(p_j^e\).
Since \(P_j\) is odd, it is invertible modulo the power of two \(M\).
Let \(\lambda_j\in\{0,\ldots,M-1\}\) be the unique integer satisfying
\[
 \lambda_j\equiv c_jP_j^{-1}\pmod M.
\]
Define
\[
 t_j=
 \begin{cases}
  \lambda_j,&\lambda_j>0,\ p_j\nmid\lambda_j,\\
  \lambda_j+M,&\text{otherwise},
 \end{cases}
\]
and
\[
 \theta_j=t_jP_j,
 \qquad
 H=\sum_{j=0}^{\ell-1}\theta_j.
\]
The definitions give
\begin{equation}\label{eq:theta-properties}
 \begin{aligned}
 &0<t_j<2M,
 &&\theta_j\equiv c_j\pmod M,\\
 &p_h^e\mid\theta_j\quad(0\le h<\ell,\ h\ne j),
 &&p_j\nmid\theta_j.
 \end{aligned}
\end{equation}
Indeed, the congruence follows from \(\theta_j=t_jP_j\), the divisibility for \(h\ne j\) follows from \(p_h^e\mid P_j\), and \(p_j\) divides neither \(t_j\) nor \(P_j\).

We next prove \(2H<K\), which ensures that two integers in \([-H,H]\) cannot be congruent modulo \(K\) unless they are equal.
Since \(\ell\ge1\), one has \(8\ell\le4^{\ell+1}\); since \(m+\ell\ge2\),
\[
 4\ell M=8\ell\,4^m
 \le4^{m+\ell+1}
 <9^{m+\ell}
 =3^e.
\]
Since every \(p_j\ge3\), \eqref{eq:theta-properties} gives
\[
 \theta_j<2MP_j\le\frac{2MK}{3^e}.
\]
Summing over \(0\le j<\ell\) yields
\begin{equation}\label{eq:h-less-than-k}
 2H<\frac{4\ell MK}{3^e}<K.
\end{equation}

Introducing one EPAD sign variable for every coordinate of \(\vec\alpha\) would make the succinct reduction exponential.
The Chinese-remainder construction instead packs all signs into the single integer \(X\).
The strict bound \(2H<K\) makes the residue-wise reconstruction unique.

\begin{lemma}[Recovering the sign vector from \(X\)]
\label{lem:recover-signs}
For every integer \(X\), the following are equivalent:
\begin{enumerate}[(i)]
\item \(|X|\le H\) and \(K\mid H^2-X^2\).
\item there exists \(\vec\alpha\in\{-1,+1\}^{\ell}\) such that
\[
 X=\sum_{j=0}^{\ell-1}\alpha_j\theta_j.
\]
\end{enumerate}
\end{lemma}
\begin{proof}
Assume (ii).
For each \(j\),
\[
 H-X=\sum_{h=0}^{\ell-1}(1-\alpha_h)\theta_h,
 \qquad
 H+X=\sum_{h=0}^{\ell-1}(1+\alpha_h)\theta_h.
\]
By \eqref{eq:theta-properties},
\[
 \alpha_j=+1\ \Longrightarrow\ p_j^e\mid H-X,
 \qquad
 \alpha_j=-1\ \Longrightarrow\ p_j^e\mid H+X.
\]
The prime powers are pairwise coprime, so \(K\mid(H-X)(H+X)=H^2-X^2\), and the triangle inequality gives \(|X|\le H\).

Conversely, assume (i).
For every \(j\),
\[
 H\equiv\theta_j\not\equiv0\pmod{p_j},
\]
because \(p_j\) divides every \(\theta_h\) with \(h\ne j\), but not \(\theta_j\).
Hence \(p_j\) cannot divide both \(H-X\) and \(H+X\): otherwise it would divide \(2H\), and therefore \(H\), since \(p_j\) is odd.
Because \(p_j^e\mid(H-X)(H+X)\), the whole power \(p_j^e\) divides exactly one factor.
Define
\[
 \alpha_j=
 \begin{cases}
 +1,&p_j^e\mid H-X,\\
 -1,&p_j^e\mid H+X.
 \end{cases}
\]
Let \(X'=\sum_{h=0}^{\ell-1}\alpha_h\theta_h\).
For every \(j\), \eqref{eq:theta-properties} gives \(X\equiv X'\pmod{p_j^e}\), hence \(X\equiv X'\pmod K\).
Both integers lie in \([-H,H]\), and \(2H<K\), so \(X=X'\).
\end{proof}

\subsection{The One-Integer Normal Form}
\label{subsec:one-integer-normal-form}

For \(\vec\alpha\in\{-1,+1\}^{\ell}\), put
\[
 X=\sum_{j=0}^{\ell-1}\alpha_j\theta_j.
\]
By \eqref{eq:theta-properties} and \eqref{eq:exact-two-linear-congruence},
\[
 X\equiv\sum_{j=0}^{\ell-1}\alpha_jc_j
 =\mathcal E(\vec\alpha)\pmod M.
\]
Thus the clause congruence becomes simply \(M\mid X\).

\begin{proposition}[One-integer encoding]
\label{prop:one-witness-exact-two}
The formula \(F_{\mathcal I}\) is satisfiable if and only if there exists \(X\in\Z\) such that
\[
 -H\le X\le H,
 \qquad
 M\mid X,
 \qquad
 K\mid H^2-X^2.
\]
\end{proposition}
\begin{proof}
Suppose first that \(F_{\mathcal I}\) is satisfiable.
By \cref{lem:exact-two-congruence}, choose \(\vec\alpha\in\{-1,+1\}^{\ell}\) with \(\mathcal E(\vec\alpha)\equiv0\pmod M\), and set \(X=\sum_{j=0}^{\ell-1}\alpha_j\theta_j\).
Then \(-H\le X\le H\) and \(K\mid H^2-X^2\) by \cref{lem:recover-signs}, while \(X\equiv\mathcal E(\vec\alpha)\equiv0\pmod M\).

Conversely, suppose \(X\) satisfies the three displayed conditions.
By \cref{lem:recover-signs}, there is \(\vec\alpha\in\{-1,+1\}^{\ell}\) with \(X=\sum_{j=0}^{\ell-1}\alpha_j\theta_j\).
Then
\[
 \mathcal E(\vec\alpha)
 \equiv\sum_{j=0}^{\ell-1}c_j\alpha_j
 \equiv\sum_{j=0}^{\ell-1}\theta_j\alpha_j
 =X
 \equiv0\pmod M.
\]
By \cref{lem:exact-two-congruence}, \(F_{\mathcal I}\) is satisfiable.
\end{proof}

\subsection{Polynomial-Space Indexed-Bit Computation for
\texorpdfstring{\(H\) and \(K\)}{H and K}}
\label{subsec:hk-bits}

The integers \(H\) and \(K\) may have exponentially many bits, so a polynomial-time reduction cannot include their binary expansions.
By \cref{lem:indexed-bits}, it suffices to give polynomial-space procedures that, on input \(\mathcal I\) and a binary bit position, return the corresponding bit of \(H\) or \(K\), together with an exponential bound on their bit lengths.

\begin{proposition}[Polynomial-space indexed-bit computation of \(H\) and \(K\)]
\label{prop:hk-fpspace}
The maps
\[
 \mathcal I\longmapsto H(\mathcal I),
 \qquad
 \mathcal I\longmapsto K(\mathcal I),
\]
extended by zero on malformed encodings, belong to \(\FPSPACE\).
For every valid instance \(\mathcal I\), let \(N=|\mathcal I|\) and put
\[
 \mathcal D(\mathcal I)=2^{4N+20}.
\]
Then
\[
 H,K,M<2^{\mathcal D(\mathcal I)}.
\]
\end{proposition}
\begin{proof}
Clause indices, variable indices, sign-coordinate indices, prime indices, and bit positions all have \(O(N)\) bits.
Evaluating the clause-description circuit \(D\) on an addressed clause takes polynomial time in \(N\).

\proofparagraph{Indexed evaluation of uniform arithmetic circuits}
We first establish the evaluation bound used below.
Consider a binary string of length \(Q\le 2^{\poly(N)}\), presented by a \(\PSPACE\) predicate that returns an indexed bit, and a logspace-uniform circuit family of size \(Q^{O(1)}\) and depth \((\log Q)^{O(1)}\), using Boolean or unweighted threshold gates.
Recursively evaluate the requested output gate.
Gate indices, input indices, child-index counters, and threshold counters use \(O(\log Q)=\poly(N)\) bits.
For an unbounded-fan-in gate, enumerate the possible child indices.
The uniformity algorithm identifies the actual children, which are recursively evaluated and counted.
The recursion depth is polynomial in \(N\).
Recomputation affects time but not space.
Hence every indexed output bit is computable in polynomial space.

We also need inverses modulo a power of two.
Given an odd \(q\)-bit integer \(a\), start with \(u=1\pmod 2\).
If \(au\equiv1\pmod{2^b}\), put \(b'=\min(2b,q)\) and
\[
 u'=u(2-au)\bmod 2^{b'}.
\]
Then \(au'=1-(au-1)^2\equiv1\pmod{2^{b'}}\).
There are \(O(\log q)\) precision-doubling stages.
The uniform division circuits of Hesse, Allender, and Barrington \cite{HesseAllenderBarrington2002,HesseAllenderBarrington2014} perform the required multiplication and binary addition in uniform constant depth.
Subtraction is two's-complement addition, and retaining the least significant \(b'\) bits is a wiring operation.
Composing the stages gives a logspace-uniform polynomial-size threshold circuit of depth \(O(\log q)\).

\proofparagraph{Coefficient residues}
For \(0\le j<\ell\), evaluate \(D\) on clause \(k\) and count the signed occurrences of \(u_j\) to obtain \(d_{jk}\in\{-4,-3,-2,-1,0,1,2,3,4\}\).
The indexed bits of \(d_{jk}4^k\bmod M\) are computable in polynomial space because \(d_{jk}\) has constant binary length.
For fixed \(j\), give a uniform iterated-addition circuit indexed access to these \(m\) residues, indexed by \(0\le k<m\).
Retaining the least significant \(\log_2M\) bits computes \(c_j\bmod M\), and the preceding circuit-evaluation argument gives polynomial-space indexed access to that residue.

\proofparagraph{Primes and \(K\)}
Given \(j<\ell\), enumerate odd candidates beginning with \(3\), test each candidate for primality by trial division, and count primes until reaching \(p_j\).
This may take exponential time, but the candidate, divisor, and counter have polynomially many bits.
Give the uniform iterated-multiplication circuit of Hesse, Allender, and Barrington indexed access to a list containing \(e\) copies of each \(p_h\), for \(0\le h<\ell\).
The circuit-evaluation argument gives polynomial-space indexed-bit access to
\[
 K=\prod_{j=0}^{\ell-1}p_j^e.
\]

\proofparagraph{The integers \(\theta_j\) and \(H\)}
For fixed \(j\), iterated multiplication gives \(p_j^e\), and uniform division gives \(P_j=K/p_j^e\).
Take the least significant \(\log_2M\) bits of \(P_j\).
Apply the inverse construction above, multiply by \(c_j\bmod M\), and retain the least significant \(\log_2M\) bits.
The result is \(\lambda_j\in\{0,\ldots,M-1\}\) satisfying
\[
 \lambda_j\equiv c_jP_j^{-1}\pmod M.
\]
Scanning the bits of \(\lambda_j\) while maintaining its residue modulo \(p_j\) decides whether \(p_j\mid\lambda_j\), and hence determines \(t_j\).
Uniform multiplication gives \(\theta_j=t_jP_j\).  A uniform balanced tree of additions gives \(H=\sum_{h=0}^{\ell-1}\theta_h\).
Thus individual bits of \(H\) and \(K\) are computable in polynomial space.

\proofparagraph{Size bound}
Because \(m\) and \(\ell\) are written in the \(N\)-bit description \(\mathcal I\),
\[
 m,\ell\le2^N,
 \qquad
 e=2(m+\ell)\le2^{N+3}.
\]
Axler's bound on the \(j\)th prime~\cite{Axler2019} gives, after increasing the fixed additive constant if necessary,
\[
 p_j<2^{3N+12}
\]
for every \(j<\ell\).
With \(\mathcal D(\mathcal I)=2^{4N+20}\),
\[
 \log_2K
 <(2^{N+1})(2^{N+3})(3N+12)
 <\mathcal D(\mathcal I).
\]
Also \(M=2^{2m+1}<2^{\mathcal D(\mathcal I)}\), and \eqref{eq:h-less-than-k} gives \(H<K<2^{\mathcal D(\mathcal I)}\).
The indexed-bit characterization in \cref{lem:indexed-bits} completes the proof.
\end{proof}

\begin{proof}[Proof of \cref{thm:succinct-exact-two-normal-form}]
The integers \(M,H,K\) were defined above, and \cref{prop:one-witness-exact-two} proves the stated equivalence.
\end{proof}

The normal form contains the integers \(H\) and \(K\), which are specified by polynomial-space indexed-bit procedures rather than by explicit binary expansions.
In the following corollary, \cref{cor:fpspace-fepad} supplies polynomial-size functional EPAD definitions of those constants, and one quotient-sharing product expresses \((H-X)(H+X)\).

\begin{corollary}[Succinct Exact-2-in-4 representation in EPAD]
\label{cor:succinct-exact-two-epad}
Given a succinct Exact-2-in-4 instance \(\mathcal I=(D,m,\ell)\), one can construct in time polynomial in \(|\mathcal I|\) an EPAD sentence \(\Xi_{\mathcal I}\) such that
\[
 \Xi_{\mathcal I}\text{ is satisfiable}
 \quad\Longleftrightarrow\quad
 F_{\mathcal I}\text{ is satisfiable}.
\]
\end{corollary}
\begin{proof}
Put \(N=|\mathcal I|\).

\proofparagraph{Defining the constants}
By \cref{prop:hk-fpspace,cor:fpspace-fepad}, the functions mapping \(\mathcal I\) to \(H(\mathcal I)\) and \(K(\mathcal I)\) have polynomial-time uniform, polynomial-size functional EPAD definitions.
Construct those formulas for input length \(N\) and substitute the concrete bits of \(\mathcal I\) for their Boolean inputs.
This gives polynomial-size formulas with unique outputs \(H\) and \(K\).
Using signed scaling from \(1\), introduce
\[
 M=2^{2m+1}.
\]

\proofparagraph{Defining the square difference}
For the product bounds, put \(\mathcal D=2^{4N+20}\).
By \cref{prop:hk-fpspace}, \(H,K,M<2^{\mathcal D}\).
Using signed scaling from \(1\), introduce the operand bound, product bound, and scaling factor
\[
 B_1=2^{\mathcal D+1},
 \qquad
 B_2=2^{2\mathcal D+2},
 \qquad
 \Lambda=2^{2\mathcal D+4}.
\]
Quantify \(X\) with \(-H\le X\le H\), and quantify \(Y\) with \(0\le Y<B_2\).
The affine forms \(H-X\) and \(H+X\) lie in \([0,2H]\), so each is less than \(B_1\), and their product is less than \(B_2\).
Apply \cref{lem:product-gadget} to the affine forms \(H-X\), \(H+X\), and \(Y\), with operand bound \(B_1\), product bound \(B_2\), and scaling factor \(\Lambda\), thereby forcing
\[
 Y=(H-X)(H+X)=H^2-X^2.
\]
The inequality
\[
 2^{2\mathcal D+4}>2^{\mathcal D+1}+2^{2\mathcal D+2}+1
\]
verifies the hypothesis of \cref{lem:product-gadget}.
Finally impose
\[
 M\mid X,
 \qquad
 K\mid Y.
\]
Let \(\Xi_{\mathcal I}\) be the existential closure of all constraints introduced above.
By \cref{prop:one-witness-exact-two}, it is satisfiable exactly when \(F_{\mathcal I}\) is satisfiable.
Every scaling exponent has polynomial-length binary representation, and the functional definitions of \(H\) and \(K\) have polynomial size, so the construction is polynomial.
\end{proof}

\begin{remark}
The sentence \(\Xi_{\mathcal I}\) belongs to the \emph{primitive positive} fragment of EPAD: after expanding the auxiliary formulas, it is an existentially quantified conjunction of affine-linear equalities, inequalities, and divisibility atoms.
Thus EPAD satisfiability is \(\NEXP\)-hard already on this fragment.
\end{remark}

\section{NEXP-Completeness of EPAD Satisfiability}\label{sec:nexp-completeness}

The lower bound is now a direct reduction from \textsc{Succinct Exact-2-in-4-SAT}.
The normal form in \cref{thm:succinct-exact-two-normal-form} represents the exponentially many Boolean variables by one integer witness and two divisibility conditions.
The modulus \(M=2^{2m+1}\) is defined by signed scaling (\cref{lem:signed-scaling}).
For \(H\) and \(K\), \cref{prop:hk-fpspace} gives polynomial-space indexed-bit procedures, and \cref{thm:alc-fpspace,cor:fpspace-fepad} give polynomial-size functional EPAD definitions of their values.
Functionality is essential: existentially quantified variables in these definitions cannot choose spurious values for the constants.

\EpadComplete*

\begin{proof}
Membership follows from the exponential small-solution theorem of Lechner, Ouaknine, and Worrell~\cite{LechnerOuaknineWorrell2015}.

For hardness, let \(\mathcal I\) be an instance of \textsc{Succinct Exact-2-in-4-SAT}.
By \cref{cor:succinct-exact-two-epad}, the EPAD sentence \(\Xi_{\mathcal I}\) is satisfiable exactly when \(F_{\mathcal I}\) is satisfiable, and it is constructed in time polynomial in \(|\mathcal I|\).
By \cref{prop:succinct-exact-two-nexp}, this proves \(\NEXP\)-hardness.
\end{proof}

Haase, Kreutzer, Ouaknine, and Worrell translate reachability for parametric one-counter automata into EPAD satisfiability and back~\cite{HaaseKreutzerOuaknineWorrell2009}.
Both translations preserve enough of the complexity to transfer the main theorem.

The translation from automata to EPAD gives the upper bound.
For the lower bound, the reverse translation must remain polynomial when coefficients are encoded in binary.

\begin{proof}[Proof of the parametric one-counter part of
\Cref{cor:downstream-hardness}]
\proofparagraph{Upper bound}
For an automaton and two control locations, one can compute EPAD formulas of polynomial size defining the reachability relation, one formula for each guessed order of the traversed zero tests~\cite[Section~4.2]{HaaseThesis2012}.
Guessing one such formula and deciding its satisfiability is a nondeterministic polynomial-time procedure with one call to EPAD satisfiability.
The upper bound of Lechner, Ouaknine, and Worrell~\cite{LechnerOuaknineWorrell2015} therefore places reachability in \(\NEXP\).

\proofparagraph{Lower bound}
The converse translation compiles a positive Boolean combination of the atoms \(A=B\), \(A<B\), \(A\mid B\), and \(\neg(A\mid B)\) into an automaton.
Each atom has a gadget: for \(A\mid B\), for instance, the automaton loads \(B\) into its counter and subtracts \(A\) until the counter reaches zero.
Sequential composition of gadgets expresses conjunction, and branching on control locations expresses disjunction.
Each gadget begins by guessing the signs of its two terms, so these case distinctions belong to the automaton and are not made by the reduction.

Two normalizations of the given formula make this translation deterministic and polynomial for binary-encoded coefficients.
First, put the formula in negation normal form, keeping \(\neg(A\mid B)\) as an atom and rewriting negated equalities and inequalities using
\[
 \neg(A=B)\ \Longleftrightarrow\ A<B\lor B<A,
 \qquad
 \neg(A<B)\ \Longleftrightarrow\ A=B\lor B<A.
\]
Since parameters are nonnegative, replace every variable \(z\) by a difference \(z^+-z^-\) of two fresh nonnegative variables.
Second, a gadget spends one edge per occurrence of a parameter, so a coefficient written in binary would require exponentially many edges.
For each product \(cz\) occurring in the formula, introduce instead fresh variables and \(O(\len(c))\) equations that double and add along the binary expansion of \(c\), and replace \(cz\) by the last of these variables.
The resulting equisatisfiable formula has only coefficients in \(\{-1,0,1\}\), and its constants may remain binary because a single edge adds any binary constant.

Both normalizations are computable in polynomial time.
Hence \Cref{thm:epad-complete} gives \(\NEXP\)-hardness, and the upper bound above gives \(\NEXP\)-completeness.
\end{proof}

\section{Functional EPAD Definitions of Exponential-Size ALCs}
\label{sec:expalc-fepad-section}

A polynomial-size EPAD translation cannot introduce one variable for every gate of an exponential-size ALC\@.
Instead, it represents correct bit-level evaluation by a succinct Exact-2-in-4 instance.
We extend the encoding from \cref{sec:succinct-exact-two-normal-form} so that an EPAD variable can specify, as a packed binary integer, the values assigned to a designated prefix of the instance's Boolean variables.
This prefix records the circuit inputs and output, together with the reference variable \(t\) from \cref{lem:three-cnf-to-exact-two}, fixed to one.

Throughout this section, \((C_n)\) may have \(n\) Boolean inputs and polynomially many free word inputs.
Each word has a uniformly specified length at most the circuit word-width, and an input gate's local description specifies the corresponding word and bit position.
Free word inputs let the computation compose with exponentially wide integers already defined by EPAD formulas.
The original Boolean-input model is the case with no free word inputs.

\subparagraph*{External interface.}
Let \(\mathcal V_n\) be the ordered list containing the Boolean input bits, the bits of each free word from least to most significant, and the output bits from least to most significant.
A propositional formula family is an \emph{evaluation encoding} of \((C_n)\) if its first variables are \(\mathcal V_n\) and, for every assignment to them, the formula has a satisfying extension exactly when the assigned output word is the value of \(C_n\) on the assigned inputs.

\subparagraph*{Locally accessible evaluation 3CNF.}
For a given input length \(n\), write \(W=W(n)\) and consider the bounded-fan-in Boolean circuit obtained by replacing each \(W\)-bit ALC gate by a standard bit-level implementation.
Use ripple-carry addition, schoolbook multiplication truncated to the lowest \(W\) bits, and fixed rewiring for logical shifts.
These implementations have polynomial size in \(W\) and admit local gate numberings.
From the \(O(\log W)\)-bit local coordinates of an internal Boolean gate, its type and predecessors are computable in polynomial time.
Fix a polynomial \(P_{\mathrm{impl}}\) that bounds all these implementation sizes.
Allocate a block of \(P_{\mathrm{impl}}(W)\) positions to every ALC gate, padding unused positions by constant-zero gates.
The Boolean gate count is then \(S(n)P_{\mathrm{impl}}(W)\), whose binary representation is computable in polynomial time.
The numbers of Boolean variables and Tseitin clause slots are fixed sums and products of the same quantities.
Their binary representations are therefore polynomial-time computable as well.
This Boolean circuit has \(2^{\poly(n)}\) gates, but we do not construct it explicitly.
A Boolean gate is indexed by the ALC-gate index and a local position inside its block.
Using the uniformity algorithm of \cref{def:alc}, its type and predecessors are therefore recoverable in polynomial time.

The Tseitin encoding is accessed locally as well.
Allocate the same fixed number of clause slots to every Boolean gate, padding unused slots by tautologies.
Given a clause index, determine the corresponding Boolean gate and clause slot, then recover the gate and predecessor information through the local interface above.
For example, an AND gate \(g=a\wedge b\) contributes
\[
 (\neg g\vee a),\quad(\neg g\vee b),\quad(g\vee\neg a\vee\neg b).
\]
Analogous constant-size clauses handle NOT, constants, and wire copies.
Pad short clauses by repeated literals.
Copy variables tie the external inputs and outputs to the corresponding circuit wires, while constraints set the unused high bits of a short word input to zero.
Number the variables in \(\mathcal V_n\) first.

The resulting clause-access algorithm runs in polynomial time on polynomial-length indices.
By the standard circuit simulation of polynomial-time computation, this clause map has a polynomial-size Boolean circuit constructible from \(1^n\) in polynomial time.
Thus the exponentially large evaluation 3CNF has a polynomial-size clause-description circuit returning three literals per clause; the 3CNF itself is never written down.

\begin{lemma}[Exact-2-in-4 evaluation instance]
\label{lem:expalc-exact-two}
For every polynomial-time uniform exponential-size ALC family and every input length \(n\), one can construct from \(1^n\), in polynomial time, a succinct Exact-2-in-4 instance \(\mathcal I_n\) and a prefix length \(L(n)=|\mathcal V_n|+1\).
The first \(L(n)-1\) variables are the ordered external bits \(\mathcal V_n\), and variable \(u_{L(n)-1}\) is the reference variable \(t\) of \cref{lem:three-cnf-to-exact-two}.
The formula \(F_{\mathcal I_n}\wedge t\) is an evaluation encoding of \(C_n\).
\end{lemma}
\begin{proof}
Construct the clause-description circuit for the evaluation 3CNF above, then apply the local part of \cref{lem:three-cnf-to-exact-two}.
Move the reference variable \(t\) immediately after the variables in \(\mathcal V_n\), shifting the intervening variable indices by one.
This reindexing is computable by binary arithmetic and preserves polynomial-time clause access.
By \cref{lem:three-cnf-to-exact-two}, fixing \(t=1\) preserves the existence of a satisfying extension for every assignment to the external bits.
\end{proof}

The lower-bound encoding in \cref{sec:succinct-exact-two-normal-form} only asks whether the succinct formula is satisfiable.
For the circuit application we need to fix a prefix of a satisfying assignment.
The following extension records exactly the additional terms required for that interface.

\begin{lemma}[Designated-prefix extension]\label{lem:designated-prefix-extension}
Let \(\mathcal I=(D,m,\ell)\) be a succinct Exact-2-in-4 instance and let \(0\le L\le\ell\).
Fix a standard polynomial-time decodable binary encoding \(\langle\mathcal I,L\rangle\) of such pairs.
For an assignment \(\vec u\in\bits^\ell\), put
\[
 \word_L(\vec u)=\sum_{b=0}^{L-1}u_b2^b.
\]
One can construct in time polynomial in \(|\langle\mathcal I,L\rangle|\) an EPAD formula \(\Xi^{\mathrm{pre}}_{\mathcal I,L}(Z)\) such that, for every integer \(Z\) with \(0\le Z<2^L\),
\[
 \Xi^{\mathrm{pre}}_{\mathcal I,L}(Z)
 \quad\Longleftrightarrow\quad
 F_{\mathcal I}\text{ has a satisfying assignment }\vec u
 \text{ with }\word_L(\vec u)=Z.
\]
\end{lemma}
\begin{proof}
Use the signs \(\vec\alpha\) from \cref{subsec:exact-two-congruence}.
Put \(A=4^m\), and recall
\[
 \mathcal E(\vec\alpha)=\sum_{k=0}^{m-1}s_k(\vec\alpha)4^k.
\]
Define
\[
 \mathcal E_L(\vec\alpha)
 =\mathcal E(\vec\alpha)+A\sum_{b=0}^{L-1}2^b\alpha_b.
\]
For the target prefix value \(Z\), put
\[
 \tau_L(Z)=2Z-(2^L-1),
 \qquad
 M_L=A2^{L+1}=2^{2m+L+1}.
\]

For \(\vec u=\vec u(\vec\alpha)\),
\[
 \mathcal E_L(\vec\alpha)-A\tau_L(Z)
 =\mathcal E(\vec\alpha)+2A\bigl(\word_L(\vec u)-Z\bigr).
\]
Since \(|\mathcal E(\vec\alpha)|<2A\) by \eqref{eq:packed-score-bound}, this difference is zero exactly when \(\mathcal E(\vec\alpha)=0\) and \(\word_L(\vec u)=Z\).
By \cref{lem:exact-two-congruence}, these conditions hold exactly when \(\vec u\) satisfies \(F_{\mathcal I}\) and has the prescribed prefix.

The same statement can be written as one congruence.
Since \(0\le\word_L(\vec u),Z<2^L\),
\[
 \begin{aligned}
 \left|\mathcal E_L(\vec\alpha)-A\tau_L(Z)\right|
 &\le |\mathcal E(\vec\alpha)|+2A\left|\word_L(\vec u)-Z\right|\\
 &<2A+2A(2^L-1)=M_L.
 \end{aligned}
\]
Hence
\begin{equation}\label{eq:prefix-congruence}
 F_{\mathcal I}\text{ has a satisfying assignment with prefix }Z
 \quad\Longleftrightarrow\quad
 \exists\vec\alpha\in\{-1,+1\}^{\ell}\colon
 \mathcal E_L(\vec\alpha)\equiv A\tau_L(Z)\pmod{M_L}.
\end{equation}

Write
\[
 \mathcal E_L(\vec\alpha)=\sum_{j=0}^{\ell-1}c_j^{(L)}\alpha_j,
\]
where, for \(0\le i<\ell\),
\[
 c_i^{(L)}
 =c_i+
 \begin{cases}
  A2^i,&i<L,\\
  0,&i\ge L.
 \end{cases}
\]
Apply the construction of \cref{subsec:assignment-one-integer} with modulus \(M_L\), coefficients \(c_i^{(L)}\), and exponent \(e=2(m+\ell)\).
The prime product remains \(K=\prod_{j<\ell}p_j^e\).
Write \(\theta_j^{(L)}\) for the resulting weights and \(H_L=\sum_{j=0}^{\ell-1}\theta_j^{(L)}\).
The uniqueness estimate remains valid.
Since \(L\le \ell\), \(\ell\le2^{\ell-1}\), and \(m+\ell\ge2\),
\[
 4\ell M_L
 =4\ell\,4^m2^{L+1}
 \le4^{m+\ell+1}
 <9^{m+\ell}
 =3^e.
\]
Thus \(2H_L<K\), exactly as in \eqref{eq:h-less-than-k}.
The recovery proof of \cref{lem:recover-signs} therefore gives, together with \eqref{eq:prefix-congruence},
\[
 \begin{aligned}
 &F_{\mathcal I}\text{ has a satisfying assignment with prefix }Z\\
 &\qquad\Longleftrightarrow\quad
 \exists X\in\Z:\
 -H_L\le X\le H_L,\quad
 M_L\mid X-A\tau_L(Z),\quad
 K\mid H_L^2-X^2.
 \end{aligned}
\]

It remains to make this description succinct.
Let \(N=|\langle\mathcal I,L\rangle|\).
For the indexed-bit procedure of \cref{prop:hk-fpspace}, the modulus exponent is \(2m+L+1\), and coefficient \(c_i\) receives the additional term \(A2^i\) when \(i<L\).
All relevant indices and exponents have \(O(N)\)-bit representations.
The bounds from \cref{prop:hk-fpspace} hold with this value of \(N\): one has \(\ell\le2^{N+1}\), \(e\le2^{N+3}\), and
\[
 2m+L+1<2^{N+3}.
\]
Consequently the maps \(\langle\mathcal I,L\rangle\mapsto H_L\) and \(\langle\mathcal I,L\rangle\mapsto K\) belong to \(\FPSPACE\).  Put
\[
 \mathcal D=2^{4N+20};
\]
then
\[
 H_L,K,M_L<2^{\mathcal D}.
\]
By \cref{cor:fpspace-fepad}, \(H_L\) and \(K\) therefore have polynomial-size functional EPAD definitions.

In the EPAD formula, generate
\[
 U_L=2^L,
 \qquad
 M_L=2^{2m+L+1},
\]
and impose \(0\le Z<U_L\).
Introduce a target variable \(R_Z\) with
\[
 \Scale_{2m}(2Z-U_L+1,R_Z),
\]
so that \(R_Z=A\tau_L(Z)\).
As in \cref{cor:succinct-exact-two-epad}, generate the operand bound, product bound, and scaling factor
\[
 B_1=2^{\mathcal D+1},
 \qquad
 B_2=2^{2\mathcal D+2},
 \qquad
 \Lambda=2^{2\mathcal D+4},
\]
quantify \(-H_L\le X\le H_L\) and \(0\le Y<B_2\).
The affine forms \(H_L-X\) and \(H_L+X\) lie in \([0,2H_L]\), hence are both less than \(B_1\).  Apply \cref{lem:product-gadget} with operand bound \(B_1\), product bound \(B_2\), and scaling factor \(\Lambda\) to force
\[
 Y=(H_L-X)(H_L+X).
\]
Finally impose
\[
 M_L\mid X-R_Z,
 \qquad
 K\mid Y.
\]
The same inequality as in \cref{cor:succinct-exact-two-epad} verifies \(\Lambda>B_1+B_2+1\), the hypothesis required by the quotient-sharing gadget.
Every new exponent has polynomial-length binary representation, so this gives the required polynomial-size formula \(\Xi^{\mathrm{pre}}_{\mathcal I,L}(Z)\).
\end{proof}

The encoding expresses correctness of an exponentially large circuit through the same succinct Exact-2-in-4 interface used by the lower-bound reduction.
Because the external bits occupy the first indices, the designated-prefix extension connects the circuit computation to free EPAD variables through one packed integer.

\begin{theorem}[EPAD representation of exponential-size ALCs]\label{thm:expalc-fepad}
Let \((C_n)\) be a polynomial-time uniform exponential-size ALC family with Boolean inputs \(\vec x=(x_0,\ldots,x_{n-1})\) and polynomially many free word inputs \(z_1,\ldots,z_t\), each of a uniformly specified length at most \(W(n)\).  Write \(\vec z=(z_1,\ldots,z_t)\).
One can construct in polynomial time an EPAD formula
\[
 \Phi_n(\vec x,\vec z,y)
\]
that is functional in \(y\) and satisfies
\[
 \Phi_n(\vec x,\vec z,y)
 \quad\Longleftrightarrow\quad
 y=C_n(\vec x,\vec z)
\]
whenever each word input lies in its declared range and \(0\le y<2^{W(n)}\).
\end{theorem}
\begin{proof}
Use \cref{lem:expalc-exact-two} to construct the succinct evaluation instance \(\mathcal I_n\).
Let the declared lengths of the word inputs be \(L_1,\ldots,L_t\), and define their starting offsets and the output offset by
\[
 o_h=n+\sum_{j=1}^{h-1}L_j\quad(1\le h\le t),
 \qquad
 o_y=n+\sum_{j=1}^tL_j.
\]
The prefix length returned by \cref{lem:expalc-exact-two} is
\[
 L=L(n)=o_y+W(n)+1.
\]
Its first \(L-1\) variables are precisely the external bits \(\mathcal V_n\), and variable \(u_{L-1}\) is the reference bit.
In EPAD, introduce \(Z\) and impose
\[
 Z=\sum_{i=0}^{n-1}2^ix_i
   +\sum_{h=1}^t2^{o_h}z_h
   +2^{o_y}y+2^{L-1}
\]
using signed scaling and affine addition.
The summands occupy disjoint bit intervals, so \(Z\) encodes the external bits followed by the reference bit \(1\).
Constrain \(0\le x_i\le1\) for every Boolean input, and constrain every word input and \(y\) to its specified range.
Apply \cref{lem:designated-prefix-extension} to \((\mathcal I_n,L)\) and conjoin the resulting formula \(\Xi^{\mathrm{pre}}_{\mathcal I_n,L}(Z)\).
The first \(L\) variables of a satisfying assignment are then exactly the bitwise concatenation of \(\vec x,\vec z,y\), followed by the reference bit \(u_{L-1}=1\).
By \cref{lem:expalc-exact-two}, satisfiability is equivalent to \(y=C_n(\vec x,\vec z)\).
Determinism of \(C_n\) gives functionality in \(y\).
\end{proof}

Free word inputs allow an EPAD formula to define an exponentially wide value and pass it to an exponential-size ALC without exposing its individual bits in the surrounding formula.

\begin{corollary}\label{cor:fexp-fepad}
One has \(\FEXP\subseteq\FEPAD\).
\end{corollary}
\begin{proof}
Combine \cref{prop:expalc-fexp,thm:expalc-fepad} with no free word inputs.
\end{proof}

\section{Finite-Quotient Normalization and Coefficient Growth}
\label{sec:normalization}

Finite-quotient normalization replaces a divisibility atom by equations for its possible integer quotients.
A \emph{normalization node} consists of the original nondivisibility constraints, the quotient equations chosen so far, and the divisibility atoms not yet replaced.
The original system is the initial node.
A \emph{branch} is a path of normalization nodes starting at the initial node.
A node is \emph{inconsistent} if its constraints have no valuation in the ambient domain.
A polynomial-size family below shows that a specified replacement sequence may produce coefficients with exponentially many bits.

\begin{definition}[Finite-quotient replacement]
\label{def:finite-quotient-replacement}
Let a normalization node contain an atom \(L\mid M\), and let \(\Gamma\) denote the conjunction of its other constraints.
Suppose there is a finite nonempty set \(\mathcal Q\subseteq\Z\) such that every valuation satisfying
\[
  \Gamma\wedge L\mid M\wedge L\ne0
\]
has \(M/L\in\mathcal Q\).
A finite-quotient replacement deletes \(L\mid M\) and, for each \(q\in\mathcal Q\), creates the node obtained by adjoining \(M=qL\) to \(\Gamma\).
These nodes contain every valuation of the original node with \(L\ne0\).
A valuation with \(L=M=0\) belongs to every resulting node because it satisfies \(0=q\cdot0\).

For each quotient equation, move all terms to the left.
If all coefficients, including the constant term, are zero, discard the equation.
If all variable coefficients are zero but the constant term is nonzero, the resulting node is inconsistent and is discarded.
Otherwise divide all coefficients by their positive greatest common divisor.
The resulting equation is \emph{primitive}.
\end{definition}

The common-zero statement is conditional and does not assert that such a valuation exists.
It is necessary because the paper uses the convention \(0\mid0\), and every quotient equation preserves any such valuation.
Finite-quotient replacement applies to affine forms and does not require homogeneity.
For homogeneous systems, quotient equations relate variable ratios rather than introduce fixed offsets.
Dividing an equation by the gcd of its coefficients preserves its integer solutions.
It also makes coefficient size independent of arbitrary rescaling.

A system \emph{reduces to divisibility-free leaves} if some finite tree of finite-quotient replacements has every leaf either inconsistent or without any divisibility atom.

To translate unrestricted EPAD formulas into systems admitting finite-quotient elimination, we use the following form of the small-solution bound of Lechner, Ouaknine, and Worrell~\cite[Theorem~14]{LechnerOuaknineWorrell2015}.

\begin{theorem}[Effective small-solution bound]
\label{thm:effective-small-solution}
There is a computable polynomial \(p\) such that every satisfiable conjunction of affine equalities, strict affine inequalities, and affine divisibility atoms of size \(n\) has an integer solution in which every variable has bit length at most \(2^{p(n)}\).
\end{theorem}

The bound lets a satisfiability-preserving reduction restrict each source variable to a finite interval.
Encoding those bounded values as quotients then gives finite ranges for the quotient choices.

\subsection{The Merge-Absorptive Fragment}\label{subsec:merge-absorptive}

D{\'e}fossez, Haase, Mansutti, and P{\'e}rez study divisibility systems with the following property.
The variables admit an order such that, in every nontrivial divisibility \(L\mid M\) entailed by the system, the greatest variable occurring in \(L\mid M\) occurs only in \(M\)~\cite{DefossezHaaseMansuttiPerez2024}.
For instance, \(x<y\wedge x+1\mid y-2\) has the property under \(x\prec y\).
Adding \(x+1\mid x+y\) destroys it, because the two atoms together entail \(x+1\mid x+2\), whose greatest variable \(x\) occurs on the left.
For the systems considered there, feasibility under this restriction is in \(\NP\).
Our fragment instead follows atoms whose finite-quotient replacements fix the ratio between two current variable components.
Each such step joins those components.
A system is merge-absorptive when these joins lead from singleton components to one component.

Recall from \cref{sec:divisibility-compressor} that a positive homogeneous divisibility system consists of divisibility atoms between homogeneous integer linear forms, together with homogeneous linear side conditions and positivity constraints.
For such a system, write \(\Phi\) for the conjunction of its homogeneous linear side conditions.
For purposes of finite-quotient elimination, we also regard a homogeneous equality \(F=0\) as the atom \(0\mid F\).
Repeated identical atoms are immaterial and may be identified.

Let \(\Gamma\) be a satisfiable conjunction of \(\Phi\) with quotient equations chosen at earlier elimination steps.
For a nonempty set \(C\) of variables, say that \(\Gamma\) \emph{fixes the ratios in \(C\)} if there are a variable \(x_C\in C\) and positive rationals \(\rho_x\), \(x\in C\), such that
\[
  \Gamma\models x=\rho_x x_C
  \qquad\text{for every }x\in C.
\]
In that case we call \(C\) a \emph{component under \(\Gamma\)}.
Thus a component is exactly a set of variables whose values are fixed positive rational multiples of one common scale.
A partition is \emph{resolved under \(\Gamma\)} if every one of its blocks is a component under \(\Gamma\).

Let \(A,B\) be distinct blocks of a partition resolved under \(\Gamma\).
An atom \(L\mid M\) \emph{merges \(A\) and \(B\) under \(\Gamma\)} if its variables lie in \(A\cup B\), it uses variables from both blocks, and there is a finite nonempty set \(\mathcal Q\subseteq\Z\) such that:
\begin{enumerate}
\item every positive integer valuation satisfying
\[
  \Gamma\wedge L\mid M\wedge L\ne0
\]
has \(M/L\in\mathcal Q\); and
\item for every \(q\in\mathcal Q\), if \(\Gamma\wedge M=qL\) is satisfiable, then \(\Gamma\wedge M=qL\) fixes the ratios in \(A\cup B\).
\end{enumerate}
The second condition also handles the common-zero case: a valuation with \(L=M=0\) satisfies every quotient equation \(M=qL\).

\begin{definition}[Merge-absorptive systems]
\label{def:merge-absorptive}
Consider a positive homogeneous divisibility system with \(r\) variables, and let \(\Phi\) be the conjunction of its homogeneous linear side conditions.
Put \(\Pi_0\) equal to the singleton partition.
The system is \emph{merge-absorptive} if there are distinct atoms
\[
  a_j=L_j\mid M_j
  \qquad(1\le j<r)
\]
and, for each \(j\), two prescribed blocks \(A_j,B_j\) of \(\Pi_{j-1}\), where
\[
  \Pi_j=
  (\Pi_{j-1}\setminus\{A_j,B_j\})
  \cup\{A_j\cup B_j\},
\]
such that the following recursive condition holds.

At stage \(j\), suppose quotients \(q_1,\ldots,q_{j-1}\) have been chosen from the finite quotient sets supplied at the preceding stages, and suppose
\[
  \Gamma_j
  =
  \Phi\wedge
  \bigwedge_{i<j}(M_i=q_iL_i)
\]
is satisfiable.
Then \(\Pi_{j-1}\) is resolved under \(\Gamma_j\), and \(a_j\) merges \(A_j\) and \(B_j\) under \(\Gamma_j\).
Choose any finite set
\[
  \mathcal Q_j(q_1,\ldots,q_{j-1})
\]
witnessing this merge.
At the next stage, \(q_j\) may be any element of this set for which
\[
  \Gamma_j\wedge M_j=q_jL_j
\]
is satisfiable.

Thus the atoms \(a_j\) and the pair \(A_j,B_j\) prescribed at each stage are fixed independently of the quotient choices.
Only the finite quotient sets and the chosen quotient values may depend on the earlier choices.
After \(r-1\) stages, \(\Pi_{r-1}\) has one block.
For \(r=1\), the empty sequence witnesses the property.
\end{definition}

This is the semantic content of merge-absorptivity: every consistent sequence of earlier quotient choices leaves the next prescribed atom able to eliminate one independent scale.
No syntactic shape for a merge atom is part of the definition.

\begin{example}[A finite-quotient replacement]\label{ex:finite-replacement}
Under the side conditions \(0<x\) and \(x\le y\le3x\), the quotient \(y/x\) belongs to \(\{1,2,3\}\).
The corresponding finite-quotient replacement has the three children shown below.

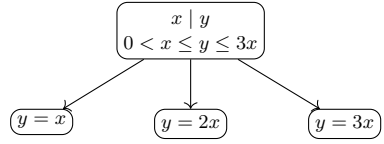
\begin{figure}[H]
\centering
\begin{tikzpicture}[scale=0.82, transform shape,
  node distance=8mm and 7mm,
  every node/.style={font=\small},
  box/.style={draw, rounded corners, align=center, inner sep=3pt}]
  \node[box] (root) {\(x\mid y\)\\ \(0<x\le y\le3x\)};
  \node[box, below left=of root] (c1) {\(y=x\)};
  \node[box, below=of root] (c2) {\(y=2x\)};
  \node[box, below right=of root] (c3) {\(y=3x\)};
  \draw[->] (root) -- (c1);
  \draw[->] (root) -- (c2);
  \draw[->] (root) -- (c3);
\end{tikzpicture}
\caption{Finite-quotient replacement branches once for each feasible value of the quotient \(y/x\).}
\label{fig:finite-quotient-replacement}
\end{figure}

The three child branches add, respectively, \(x-y=0\), \(2x-y=0\), and \(3x-y=0\).
Equivalently, they impose \(y=x\), \(y=2x\), and \(y=3x\).
Each child replaces the divisibility atom by one possible quotient equation.\exampleqed
\end{example}

\begin{example}[A two-step sequence of merge atoms]\label{ex:merge-sequence}
Take variables \(a,b,c\) over \(\Npos\), side conditions \(a>b>c\), and the two divisibility atoms \(2a+b\mid 5a\) and \(a\mid b+3c\).

The first atom merges \(\{a\}\) and \(\{b\}\): since \(a>b\), the quotient \(5a/(2a+b)\) lies strictly between \(5/3\) and \(5/2\), so it equals \(2\) and forces \(a=2b\).
The second atom then merges \(\{c\}\) into \(\{a,b\}\): using \(a=2b\) and \(b>c\), the quotient \((b+3c)/a\) lies strictly between \(1/2\) and \(2\), so it equals \(1\) and forces \(b=3c\).
Thus all variables lie in one component, and the solutions that remain form the one-parameter set \((a,b,c)=(6s,3s,s)\), \(s\in\Npos\).

\begin{figure}[H]
\centering
\begin{tikzpicture}[
  node distance=7mm,
  every node/.style={font=\small},
  part/.style={draw, rounded corners, inner sep=3pt, minimum width=10mm}]
  \node[part] (a0) {\(\{a\}\)};
  \node[part, right=7mm of a0] (b0) {\(\{b\}\)};
  \node[part, right=7mm of b0] (c0) {\(\{c\}\)};
  \node[part, below=10mm of b0, xshift=-8mm] (ab) {\(\{a,b\}\)};
  \node[part, right=7mm of ab] (c1) {\(\{c\}\)};
  \node[part, below=10mm of ab, xshift=8mm] (abc) {\(\{a,b,c\}\)};
  \draw[->] (b0.south) --
    node[right, font=\scriptsize] {\(2a+b\mid5a\)} (ab.north);
  \draw[->] (ab.south) --
    node[left, font=\scriptsize] {\(a\mid b+3c\)} (abc.north west);
  \draw[->] (c1.south) -- (abc.north east);
\end{tikzpicture}
\caption{The two bounded-quotient atoms successively merge three singleton variable components into one component.}
\label{fig:two-step-merge}
\end{figure}
\exampleqed
\end{example}

\begin{remark}[Comparison with ordered divisibility systems]
\label{rem:ordered-vs-merge}
The ordered restriction described above and merge-absorptivity are incomparable.
The system
\[
 x\mid y,\qquad x>0,\quad y>0,
\]
satisfies the ordered restriction for \(x\prec y\), but is not merge-absorptive because the quotient \(y/x\) is unbounded.
Conversely, \(\Delta_1\) is merge-absorptive: under \(0<y_0<y_1\), the quotient in its atom \(y_1+y_0\mid4y_1\) lies strictly between \(2\) and \(4\), so it is \(3\) and forces \(y_1=3y_0\).
It does not satisfy the ordered restriction.
In its atom
\[
 y_1-2y_0\mid y_1,
\]
the order \(y_0\prec y_1\) places the greatest variable on both sides, whereas the order \(y_1\prec y_0\) places it only on the left.
\end{remark}

A merge atom supplies one locally finite quotient choice.
Merge-absorptivity requires a compatible sequence of such choices that contracts the entire variable partition.
The scaling gadgets satisfy this global condition.

\begin{lemma}[Scaling systems are merge-absorptive]
\label{lem:scaling-systems-merge-absorptive}
The systems in \cref{lem:mersenne-multiplier,lem:positive-scaling} are merge-absorptive.
\end{lemma}
\begin{proof}
For \(t=0\), the equality \(r=x\) fixes the relative scale of \(r\) and \(x\).
Let \(t\ge1\), and process the atoms~\eqref{eq:stage-merge} for \(i=t,t-1,\ldots,1\).

Fix any consistent choices of quotients made at the earlier stages.
By induction, the resulting quotient equations fix \(y_t\) and \(y_i\) as
\[
  y_t=\alpha s,\qquad y_i=\beta s
\]
for positive rationals \(\alpha,\beta\) and one positive scale \(s\).
The variable \(y_{i-1}\) is still a singleton component.
For
\[
  y_i+y_{i-1}\mid y_t+3y_i,
\]
the divisor is greater than \(\beta s\), while the dividend is \((\alpha+3\beta)s\).
Hence every possible positive integer quotient \(q\) satisfies
\[
  0<q<\frac{\alpha+3\beta}{\beta},
\]
so there are only finitely many possibilities.
For any \(q\) whose quotient equation is consistent,
\[
  (\alpha+3\beta)s=q(\beta s+y_{i-1})
\]
fixes \(y_{i-1}\) as a positive rational multiple of \(s\).
Thus this atom merges the component containing \(y_t,y_i\) with \(\{y_{i-1}\}\), independently of the earlier quotient choices.

A positive-scaling chain consists of these systems joined by homogeneous equalities.
Once the variables on one side of such an equality form one component, the equality fixes the next variable relative to the same scale.
\end{proof}

\begin{proposition}[Decidability and normalization]
\label{prop:merge-absorptive-recognition}
For positive homogeneous divisibility systems, merge-absorptivity is decidable.
Moreover, every merge-absorptive system reduces to divisibility-free leaves by finite-quotient replacement.
\end{proposition}
\begin{proof}
Fix a positive homogeneous divisibility system with \(r\) variables, and let \(\Phi\) be the conjunction of its homogeneous linear side conditions.
There are finitely many possible sequences of \(r-1\) distinct atoms and finitely many possible sequences of component unions.
Fix one such candidate sequence, and write its selected atoms as \(a_i=L_i\mid M_i\) for \(1\le i<r\).
We decide whether it satisfies \cref{def:merge-absorptive} by considering recursively every consistent sequence of quotient choices generated by it.

Suppose that at stage \(j\) the earlier choices \(q_1,\ldots,q_{j-1}\) give the satisfiable homogeneous linear conjunction
\[
  \Gamma_j=
  \Phi\wedge\bigwedge_{i<j}(M_i=q_iL_i).
\]
Whether \(\Gamma_j\) fixes the ratios in any proposed component is decidable by rational linear programming.
Assume therefore that the current partition is resolved, and consider the prescribed blocks \(A,B\).
Write every variable of \(A\) as a fixed positive rational multiple of a scale \(t_A\), and every variable of \(B\) as a fixed positive rational multiple of a scale \(t_B\).
Write the selected forms as
\[
  L_j=\alpha t_A+\beta t_B,
  \qquad
  M_j=\gamma t_A+\delta t_B
\]
for rational \(\alpha,\beta,\gamma,\delta\).

Project the positive rational solution set of \(\Gamma_j\) onto \((t_A,t_B)\), and normalize \(t_A=1\).
The possible ratios \(z=t_B/t_A\) form a rational interval \(I\subseteq\mathbb R_{>0}\), possibly open or unbounded, whose endpoints are computable by linear programming.
Homogeneity implies that rational feasibility is equivalent to positive-integer feasibility after scaling by a common denominator.

When the divisor is nonzero, the quotient \(M_j/L_j\) is
\[
  \frac{\gamma+\delta z}{\alpha+\beta z}.
\]
Its possible integer values are finite exactly when this linear-fractional function does not approach infinity on the feasible interval.
The denominator has at most one zero \(z_0\).
If \(z_0\) is not in the closure of \(I\), the quotient is bounded.
If \(z_0\) is in the closure of \(I\) and the numerator is nonzero at \(z_0\), then solving
\[
  \gamma+\delta z=q(\alpha+\beta z)
\]
shows that, for all sufficiently large integers \(q\) of one sign, the resulting rational value of \(z\) lies in \(I\); hence infinitely many integer quotients occur.
If numerator and denominator vanish at the same \(z_0\), the two affine functions are proportional away from the degenerate identically-zero cases, so the nonzero-divisor quotient is constant.
The identically-zero cases are checked directly.
Thus we can decide whether a finite quotient set exists and, when it does, compute the possible nonzero-divisor quotients.

For each such quotient \(q\), test the satisfiability of
\[
  \Gamma_j\wedge M_j=qL_j.
\]
If it is satisfiable, normalize \(t_A+t_B=1\) and minimize and maximize \(t_A\).
The conjunction fixes the ratios in \(A\cup B\) exactly when these two extrema agree.
This decides whether the selected atom merges \(A,B\) under \(\Gamma_j\).
Since each successful stage has finitely many quotient choices and there are only \(r-1\) stages, the recursive test terminates.
Enumerating the finitely many candidate atom and component sequences decides merge-absorptivity.

For the systems constructed in this paper, three sufficient conditions simplify the general test and make the concrete gadgets easier to check.
For a component \(C\), call a valuation \emph{admissible} if it is positive and satisfies the side conditions supported on \(C\).
A form supported on \(C\) is \emph{nonvanishing} on \(C\) if it is nonzero on every admissible valuation; positivity and negativity are defined analogously.
For disjoint components \(A,B\), say that a form \(R_B\) supported on \(B\) is \emph{side-condition bounded by \(A\)} if every variable in its support is bounded above on the real polyhedron given by the side conditions on \(A\cup B\), positivity, and the normalization
\[
  \sum_{x\in A}x=1.
\]
For an atom \(L\mid M\) supported on \(A\cup B\), after possibly exchanging \(A\) and \(B\), each of the following guarantees a semantic merge:
\begin{enumerate}
\item \emph{Linear merge:} \(L=0\) and \(M=F_A+F_B\), where \(F_A\) and \(F_B\) are nonvanishing on their components.
\item \emph{Bounded-divisor merge:} \(L=P_A+R_B\) and \(M=Q_A\), where \(P_A,R_B\) have the same strict sign on their components and \(Q_A\) is nonvanishing on \(A\).
\item \emph{Bounded-dividend merge:} \(L=P_A\) and \(M=T_A+R_B\), where \(P_A\) is nonvanishing on \(A\), \(R_B\) is nonvanishing on \(B\), and \(R_B\) is side-condition bounded by \(A\).
\end{enumerate}
These conditions are decidable by linear programming.
The first fixes the relative scale directly.
For the second, after resolving the components write
\[
  P_A=\lambda_P t_A,\qquad R_B=\lambda_R t_B,\qquad Q_A=\lambda_Q t_A.
\]
The sign hypotheses give \(\lambda_P,\lambda_R\ne0\) with the same sign and \(\lambda_Q\ne0\), so
\[
  \left|\frac{Q_A}{P_A+R_B}\right|
  \le \left|\frac{\lambda_Q}{\lambda_P}\right|.
\]
Hence there are finitely many quotients, and every consistent quotient equation fixes \(t_B/t_A\).
For the third, write
\[
  P_A=\lambda_P t_A,\qquad T_A=\lambda_T t_A,\qquad R_B=\lambda_R t_B.
\]
Side-condition boundedness bounds \(t_B/t_A\), and therefore bounds \((T_A+R_B)/P_A\); every consistent quotient equation again fixes the ratio of the two scales.
These are convenient certificates for the semantic definition, not part of that definition.

Finally, follow the prescribed merge sequence of a merge-absorptive system.
For every consistent sequence of quotient choices through all \(r-1\) stages, all variables are fixed rational multiples of one positive parameter \(t\).
Every remaining divisibility is therefore of the form
\[
  \lambda t\mid\mu t
\]
for rational constants \(\lambda,\mu\).
After clearing denominators and cancelling \(t>0\), this is a constant divisibility test.
Thus every terminal normalization node is either inconsistent or can have all remaining divisibility atoms removed.
\end{proof}

The structural result applies to positive homogeneous systems.
For completeness of EPAD, the following propositions first remove Boolean structure and signs, then homogenize the result while arranging an explicit sequence of component merges.

An \emph{affine divisibility system} is a conjunction of affine-linear equalities, strict affine-linear inequalities, and divisibility atoms between affine-linear terms over integer variables.

\begin{proposition}[EPAD to affine divisibility]\label{prop:epad-affdivz}
Every EPAD formula has a polynomial-time equisatisfiable translation to an affine divisibility system over integer variables.
\end{proposition}
\begin{proof}
Let \(n\) be the size of the input formula.
Put the quantifier-free matrix in negation normal form.
Negated affine atoms become positive Boolean combinations of affine equalities and strict affine inequalities.
Negated divisibility atoms can also be rewritten locally as in \cite[Section~IV.A]{LechnerOuaknineWorrell2015}.
Thus, after adding only polynomially many arithmetic variables, the formula is a positive Boolean combination of affine equalities, strict affine inequalities, and unnegated affine divisibility atoms.

We first bound all arithmetic variables.
If the positive Boolean matrix is satisfiable, choose a conjunction of atoms along a satisfied branch.
By \cref{thm:effective-small-solution}, that conjunction has a bounded solution.
Because the matrix is positive, the same solution satisfies the whole matrix.
Let \(p\) be the polynomial from \cref{thm:effective-small-solution}, increased if necessary to absorb the polynomial-size preprocessing above.  Introduce a variable \(B\) and force \(B=2^{2^{p(n)}+1}\) by \(\Scale_{2^{p(n)}+1}(1,B)\).
Add \(-B<x<B\) for every original and auxiliary arithmetic variable \(x\).
For every affine term and product variable introduced in the construction, add a variable constrained to equal a power of two that bounds its absolute value.
The exponents have polynomial bit length, so all scaling systems are polynomial-size conjunctions of affine divisibility constraints and strict affine inequalities.

It remains to remove the Boolean structure.
Give each subformula \(\varphi\) a Boolean variable \(b_\varphi\), constrained by the strict inequalities \(b_\varphi+1>0\) and \(2-b_\varphi>0\).
For a conjunction or disjunction with immediate subformulas \(\psi\) and \(\theta\), introduce a product variable \(w\), add \(-2<w<2\), and apply \Cref{lem:product-gadget} with \(U=b_\psi\), \(V=b_\theta\), \(W=w\), \(M_U=M_W=2\), and the constant \(\Lambda=8>2+2+1\).
This enforces \(w=b_\psi b_\theta\).
Add \(b_{\psi\wedge\theta}=w\) or \(b_{\psi\vee\theta}=b_\psi+b_\theta-w\), according to the connective.
For an atomic formula \(\alpha\), it is enough to enforce \(b_\alpha=1\Rightarrow\alpha\).
Because the syntax tree is positive and the variable of the whole matrix is forced to be \(1\), these one-way implications force a satisfying assignment for the original formula.
Reverse implications are unnecessary for equisatisfiability.
If \(\alpha\) is an equality \(h=0\), introduce \(T_\alpha=b_\alpha h\) and require \(T_\alpha=0\).
If \(\alpha\) is a strict inequality \(h>0\), introduce \(T_\alpha=b_\alpha h\) and require \(T_\alpha-b_\alpha+1>0\).
If \(\alpha\) is a divisibility \(L\mid M\), introduce \(L_\alpha=b_\alpha L\) and \(M_\alpha=b_\alpha M\), and require \(L_\alpha\mid M_\alpha\).

For each product variable among \(T_\alpha,L_\alpha,M_\alpha\) that is introduced above, the bounds fixed above provide an explicit power-of-two bound on its absolute value.
Add that bound as side conditions and use it as \(M_W\) in the corresponding application of \cref{lem:product-gadget}; take \(U=b_\alpha\), let \(V\) be the affine form \(h\), \(L\), or \(M\) being gated, and let \(W\) be its already named product variable.
Take \(M_U=2\) and choose the gadget's power of two \(\Lambda>M_U+M_W+1\).
The exponent of \(\Lambda\) has polynomial bit length, and every multiple by \(\Lambda\) is represented by fresh variables constrained through \Cref{lem:signed-scaling}.
Require the Boolean variable associated with the whole matrix to be \(1\).

If the original EPAD formula is satisfiable, choose a bounded satisfying assignment, set each Boolean variable to the truth value of its subformula, and interpret the product variables as the corresponding products.
All constraints above are then satisfied.
Conversely, a satisfying assignment of the affine-divisibility conjunction makes the Boolean output \(1\).
Since every atom marked true is actually true, the original positive Boolean matrix is true under the arithmetic assignment.
Hence the construction is equisatisfiable.
\end{proof}

The affine conjunction need not satisfy the contraction property.
The next translation represents each original value by a quotient exposed along a decreasing positive chain.
The quotient atoms also supply the required merges.

\begin{proposition}[Merge-absorptive completeness]
\label{prop:merge-absorptive-completeness}
Every EPAD formula has a polynomial-time equisatisfiable translation to a \MergeAbs{} positive homogeneous divisibility system.
\end{proposition}
\begin{proof}
By \Cref{prop:epad-affdivz}, reduce an affine divisibility system over integer variables.
Replace each \(x\) by \(x^+-x^-\), with \(x^+,x^-\in\Npos\).
Replace each affine equality \(h=0\) by \(0\mid h\), and each strict affine inequality \(h>0\) by \(0\mid h-s\) with a fresh \(s\in\Npos\).
This gives an equisatisfiable affine divisibility system \(\mathcal J\) over positive integer variables.
Let \(x_1,\ldots,x_k\) be its variables and write \(\vec x=(x_1,\ldots,x_k)\).
Write the atoms of \(\mathcal J\) as \(L(\vec x)\mid M(\vec x)\), where \(L,M\) are affine-linear forms with integer coefficients.
If rational affine forms arise, multiplying both terms of an atom by a common positive denominator gives an equivalent atom with integer coefficients.

Let \(p\) be the polynomial from \cref{thm:effective-small-solution}, and let \(\nu=2^{p(|\mathcal J|)}\) be the resulting bit-length bound.
If \(\mathcal J\) is satisfiable, then it has a solution with \(x_j<2^\nu\) for every \(j\).
Put
\[
 E=\nu+\lceil\log_2(k+1)\rceil+1.
\]
Then \(k2^\nu<2^E\), and the binary length of \(E\) is polynomial in \(|\mathcal J|\).

We construct a positive homogeneous system.
Introduce \(u\) and variables \(s_j\) for \(0\le j\le k\), with side conditions \(s_{j-1}>s_j\) for \(1\le j\le k\).
Using \cref{lem:positive-scaling}, add a scaling system forcing \(s_0=2^E u\).
For each \(j=1,\ldots,k\), add the atom \(u\mid s_{j-1}-s_j\).
Since \(s_{j-1}>s_j\) and \(u>0\), every solution determines a positive integer quotient
\[
 q_j=\frac{s_{j-1}-s_j}{u},
\]
which encodes \(x_j\).
The choice of \(E\) ensures that every bounded solution of \(\mathcal J\) can be encoded in this form.

For each affine form \(L(\vec x)=c+\sum_{j=1}^k a_jx_j\) appearing in \(\mathcal J\), define the homogeneous translated form \(\widehat L=c u+\sum_{j=1}^k a_j(s_{j-1}-s_j)\).
Thus, on every assignment satisfying the atoms \(u\mid s_{j-1}-s_j\), we have \(\widehat L=uL(q_1,\ldots,q_k)\).
For every atom \(L(\vec x)\mid M(\vec x)\) of \(\mathcal J\), add the translated atom \(\widehat L\mid\widehat M\).
The resulting system has size polynomial in \(|\mathcal J|\): the only large number used as an exponent is \(E\), whose binary length is polynomial, and \Cref{lem:positive-scaling} produces a system of size polynomial in \(\log(E+2)\).

We first check that the constructed system is \MergeAbs.
By \Cref{lem:scaling-systems-merge-absorptive}, the constraints enforcing \(s_0=2^E u\) admit a prescribed merge sequence that puts \(s_0\), \(u\), and all auxiliary variables used by those constraints into one component.
Now fix \(j\), and suppose arbitrary consistent quotients have been chosen for the atoms \(u\mid s_{h-1}-s_h\) with \(1\le h<j\).
The resulting quotient equations fix \(u=\alpha t\) and \(s_{j-1}=\beta t\) for positive rationals \(\alpha,\beta\) and one positive scale \(t\), while \(s_j\) is still a singleton.
Because \(0<s_j<s_{j-1}\),
\[
  0<
  \frac{s_{j-1}-s_j}{u}
  <
  \frac{\beta}{\alpha}.
\]
Thus \(u\mid s_{j-1}-s_j\) has only finitely many possible positive integer quotients.
For the quotient value \(q_j\) on such a branch,
\[
  s_j=(\beta-q_j\alpha)t,
\]
so \(s_j\) is fixed as a positive rational multiple of \(t\).
Hence this atom merges \(\{s_j\}\) into the current component for every consistent sequence of earlier quotient choices.
Iterating for \(j=1,\ldots,k\) merges all variables introduced above.
The translated source atoms are then supported inside the single component and need not be used in the merge sequence.

It remains to prove equisatisfiability.
Suppose first that \(\mathcal J\) has a satisfying assignment \(\xi_1,\ldots,\xi_k\) with \(\xi_j<2^\nu\).
Set \(u=1\), \(s_0=2^E\), and define \(s_j=s_{j-1}-\xi_j\) for \(1\le j\le k\).
The choice of \(E\) gives \(\xi_1+\cdots+\xi_k<2^E\), so \(s_0>s_1>\cdots>s_k>0\).
The constraints enforcing \(s_0=2^E u\) have a solution by \Cref{lem:positive-scaling}, and every atom \(u\mid s_{j-1}-s_j\) has quotient \(\xi_j\) for \(j\le k\).
For each translated atom, \(\widehat L=L(\xi_1,\ldots,\xi_k)\) and \(\widehat M=M(\xi_1,\ldots,\xi_k)\), because \(u=1\).
Therefore all translated divisibility atoms hold, and the constructed system is satisfiable.

Conversely, suppose the constructed system is satisfiable.
For \(j=1,\ldots,k\), define \(q_j=(s_{j-1}-s_j)/u\).
The atoms and the side conditions \(s_{j-1}>s_j\) imply that every \(q_j\) is a positive integer.
For every atom \(L\mid M\) of \(\mathcal J\), the translated atom gives \(u\,L(q_1,\ldots,q_k)\mid u\,M(q_1,\ldots,q_k)\).
Since \(u>0\), this is equivalent to \(L(q_1,\ldots,q_k)\mid M(q_1,\ldots,q_k)\).
Thus \(q_1,\ldots,q_k\) satisfies \(\mathcal J\).
This proves correctness of the polynomial-time many-one reduction.
\end{proof}

The construction uses the exposed quotients both to encode source values and to form the merge sequence.
Every solution recovers the source values from those quotients.

\MergeAbsHard*
\begin{proof}
Immediate from \Cref{thm:epad-complete,prop:merge-absorptive-completeness}.
\end{proof}

Finite-quotient replacement eliminates all divisibility atoms in this fragment, but the equations produced along a replacement sequence need not have polynomial-size coefficients.
The multiplier systems give an exponential example.

\subsection{A Family Forcing Exponential-Bit Coefficients}
\label{subsec:large-coefficients}

Recall that \(\Delta_j(x,r)\) is polynomial-size and \MergeAbs{}, and forces \(r=(2^{2^j}-1)x\).

\begin{lemma}[Quotient at each stage]
\label{lem:exposed-quotient}
For \(j\ge1\), consider the finite-quotient replacement sequence for \(\Delta_j\) that processes the merge atoms in the order from \Cref{lem:scaling-systems-merge-absorptive}.
At each stage \(i=j,j-1,\ldots,1\), the only satisfiable child adds the equation
\[
 y_j+3y_i=(2^{2^{j-i}}+1)(y_i+y_{i-1}).
\]
\end{lemma}
\begin{proof}
Put \(h=j-i\) and \(B=2^{2^h}\).
The stage invariant and forced ratio are
\[
 y_j=(B-1)y_i,\qquad y_i=(B+1)y_{i-1}.
\]
Consequently,
\[
 \frac{y_j+3y_i}{y_i+y_{i-1}}
 =\frac{(B+2)(B+1)y_{i-1}}{(B+2)y_{i-1}}
 =B+1.
\]
Every child with another quotient is inconsistent, because the soundness argument of \cref{lem:mersenne-multiplier} already forces \(y_i=(B+1)y_{i-1}\), which pins the quotient to \(B+1\).
The surviving equation is
\[
 y_j+(2-2^{2^h})y_i-(2^{2^h}+1)y_{i-1}=0.
\]
The coefficient of \(y_j\) is \(1\), so gcd normalization leaves this equation unchanged.
At \(i=1\), the retained coefficient has bit length \(\Theta(2^j)\).
\end{proof}

\ReplacementBlowup*
\begin{proof}
Polynomial size follows from \Cref{lem:mersenne-multiplier}, and membership in the merge-absorptive fragment follows from \Cref{lem:scaling-systems-merge-absorptive}.
By \Cref{lem:exposed-quotient}, the specified replacement sequence produces \(2^{2^{j-1}}+1\), whose bit length is \(\Theta(2^j)\).

Put \(M_j=2^{2^j}-1\).
The projected solution set is exactly
\[
 \{(x,M_jx):x\in\Npos\}.
\]
Any affine equation defining this ray and involving only \(x,r\) is nonzero.
Substituting \((x,r)=(t,M_jt)\) for all \(t\in\Npos\) shows that it is \(b(r-M_jx)=0\) for some nonzero \(b\in\Z\).
It therefore contains a coefficient of absolute value \(|b|M_j\ge M_j\), with bit length at least \(2^j\).
\end{proof}

The growth exhibited by this family is not an artifact of the family.
It is forced by the complexity of the fragment.

\begin{corollary}[No polynomial-size normalization]
\label{cor:no-poly-normalization}
No finite-quotient elimination produces a polynomial-size surviving branch on every satisfiable \MergeAbs{} positive homogeneous divisibility system.
\end{corollary}

\begin{proof}
Suppose one did.
Apply the polynomial-time reduction of \cref{prop:merge-absorptive-completeness} to an arbitrary EPAD instance.
Its output is merge-absorptive.
By \cref{def:merge-absorptive}, a system with \(r\) variables has \(r-1\) prescribed merge atoms, so a complete choice of quotients gives \(r-1\) quotient equations.
A nondeterministic machine guesses these equations.
An equation \(M=qL\) entails \(L\mid M\), so every guess is sound, while a satisfying assignment supplies a consistent choice of the actual quotients and therefore is not lost.
Reject a guess whose linear constraints are inconsistent or do not fix all variable ratios.
Otherwise every variable is a rational multiple of one positive parameter.
Every remaining divisibility atom then reduces to a constant test, as in \cref{prop:merge-absorptive-recognition}, and integer feasibility of the resulting polynomial-size linear system is decidable in \(\NP\).

If every satisfiable merge-absorptive system admitted a polynomial-size such normalization, this would give an \(\NP\) procedure for unrestricted EPAD satisfiability.
By \cref{thm:epad-complete}, that would imply \(\NEXP\subseteq\NP\), contradicting the nondeterministic time hierarchy theorem.
\end{proof}

\Cref{thm:replacement-blowup} is the explicit form of this obstruction.
It names one polynomial-size family and one replacement sequence, and its argument is arithmetic rather than complexity-theoretic.
It fixes the bit length at \(\Theta(2^j)\) and proves the same lower bound for every affine equation over the projected variables.

\clearpage

\subsection*{AI-usage Disclosure}
We used generative AI tools to polish the writing, help find counterexamples, and explore proof ideas.
We independently verified all mathematical arguments, examples, and references and take full responsibility for the contents of the paper.

\bibliography{refs}

\clearpage

\appendix

\section{Proof of the Word-Equation Corollary}
\label{sec:word-equations}

Day and Konefal show that every EPAD-definable relation can be built from length abstractions of word equations~\cite{DayKonefal2025}.
Thus their result and \cref{thm:epad-complete} give the word-equation lower bound.
For completeness, we verify here that the transfer can be written with nonerasing variables over the fixed alphabet \(\{a,b\}\) and with positive outer Boolean structure.

Target word variables denote nonempty words.
Length constraints may be arbitrary quantifier-free Presburger formulas, while the outer formula uses only conjunction and disjunction.
We represent a nonnegative integer \(v\) by a word of length \(v+1\), so arithmetic zero also has a nonempty representative.
Presburger conditions then become length constraints directly.

The remaining ingredient is the commutation theorem for words.
A nonempty word is \emph{primitive} if it is not a proper power.
If \(XY=YX\), then \(X=P^m\) and \(Y=P^n\) for a word \(P\) and integers \(m,n\ge0\).
Hence every nonempty word commuting with a primitive word \(X\) lies in \(X^+\)~\cite[Section~1.3]{Lothaire1997}.
The proof below uses a primitive word over \(\{a,b\}\) to realize divisibility of positive lengths and adds one case for a zero dividend.

\begin{proposition}[Nonerasing two-letter specialization]
\label{prop:epad-to-nonerasing-we}
EPAD formula satisfiability many-one reduces in polynomial time to nonerasing satisfiability for positive Boolean combinations of word equations and quantifier-free Presburger length constraints over the alphabet \(\{a,b\}\).
\end{proposition}

\begin{proof}
We first put the EPAD input into the form used by the reduction.
Replace every integer variable \(x\) by \(x^+-x^-\), where \(x^+\) and \(x^-\) range over \(\N\).
This is equisatisfiable and keeps all terms affine.
Push negations to atoms.
For each divisibility literal involving affine terms \(L,M\), introduce nonnegative variables \(A=|L|\) and \(B=|M|\).
The absolute-value constraints are Presburger, and \(L\mid M\) over \(\Z\) is equivalent to \(A\mid B\).
For a negated atom, use the following nonnegative specialization of the standard rewrite~\cite[Section~IV.A]{LechnerOuaknineWorrell2015}:
\[
 \neg(A\mid B)
 \ \Longleftrightarrow\
 (A=0\wedge B>0)\ \lor\
 \exists Q,R\in\N\,
 \bigl(B=Q+R\wedge A\mid Q\wedge 0<R<A\bigr).
\]
The equivalence is Euclidean division when \(A>0\), together with the convention \(0\mid0\).
It introduces only unnegated divisibility.
Thus the result is an equisatisfiable positive Boolean combination of quantifier-free Presburger conditions and unnegated divisibility atoms over \(\N\).

We represent a nonnegative arithmetic value \(v\) by a nonempty word of length \(v+1\).
For a word variable \(X\), write \(\val(X)=|X|-1\).
Each arithmetic variable \(x\) receives a word variable \(X\), and Presburger conditions are translated by replacing \(x\) with \(\val(X)\).
This is still a quantifier-free Presburger length constraint.
The subtraction in \(\val(X)\) is only notation.
It can be removed by moving constants to the other side.

It remains to translate unnegated divisibility atoms.
First consider positive lengths.
For word variables \(X,Y\), define
\[
  \Div_+(X,Y) \coloneqq
  XY=YX\ \wedge\ aX=Zb\ \wedge\ Za=aZ,
\]
where \(Z\) is a fresh auxiliary word variable.
For nonempty word values, \(\Div_+(X,Y)\) enforces exactly the relation \(|X|\mid |Y|\) on positive lengths.
The equation \(Za=aZ\) forces \(Z=a^j\) with \(j\ge1\), and then \(aX=Zb\) forces \(X=a^{j-1}b\).
This word is primitive, so \(XY=YX\) forces \(Y\) to be a positive power of \(X\).
Conversely, for positive integers \(x,y\) with \(x\mid y\), use \(Z=a^x\), \(X=a^{x-1}b\), and \(Y=X^{y/x}\).

To handle zero, define the following formula for shifted lengths:
\[
  \Div_0(X,Y) \coloneqq
  |Y|=1\ \vee\
  \bigl(|U|+1=|X|\wedge |V|+1=|Y|\wedge \Div_+(U,V)\bigr),
\]
with fresh auxiliary word variables \(U,V\) and the auxiliary variable used inside \(\Div_+\).
For nonempty \(X,Y\), this formula satisfies \(\Div_0(X,Y)\) if and only if \(\val(X)\mid\val(Y)\).
If \(\val(Y)=0\), the first disjunct applies, expressing that every nonnegative value divides zero, including the case \(0\mid0\).
If \(\val(Y)>0\) and \(\val(X)=0\), the first disjunct is false and the second would require the nonempty word \(U\) to have length zero.
Thus divisibility by zero is rejected when the dividend is positive.
If both shifted values are positive, the second disjunct gives words \(U,V\) of lengths \(\val(X),\val(Y)\), and \(\Div_+(U,V)\) is equivalent to positive divisibility.

For each unnegated divisibility atom \(\alpha=L\mid M\), introduce fresh word variables \(A_\alpha,B_\alpha\), add the length constraints \(\val(A_\alpha)=L\) and \(\val(B_\alpha)=M\), and conjoin \(\Div_0(A_\alpha,B_\alpha)\).
All auxiliary variables are interpreted existentially by the target satisfiability problem.

The construction is equisatisfiable.
From an arithmetic satisfying assignment, choose nonempty words of the prescribed shifted lengths.
The divisibility formulas are satisfied by the previous paragraphs.
Conversely, a nonerasing solution assigns each arithmetic variable \(x\) the value \(|X|-1\).
The translated Presburger constraints have the same truth values as their arithmetic originals, and each \(\Div_0\) formula enforces the corresponding divisibility atom.
The initial handling of signs and negated divisibility was logically equivalent, so the recovered arithmetic assignment satisfies the original EPAD formula.

The construction is polynomial.
Each transformation is local, each affine term is encoded by a linear-size length constraint, and each divisibility atom contributes only constantly many word-equation atoms and length constraints, apart from the size of the terms it mentions.
The only terminal letters used are \(a\) and \(b\).
\end{proof}

This direct specialization establishes the syntactic restrictions in \cref{cor:downstream-hardness}.
The underlying transfer through length abstractions is due to Day and Konefal~\cite{DayKonefal2025}.

\begin{proof}[Proof of the word-equation part of
\Cref{cor:downstream-hardness}]
Use \cref{thm:epad-complete}, then apply \cref{prop:epad-to-nonerasing-we}.
\end{proof}

\end{document}